%% file: main.tex
\documentclass{article}
\title{Learning stability guarantees for constrained switching linear systems from noisy observations\thanks{This work has been partially presented at the 61st IEEE Conference on Decision and Control \cite{cdc22}. More information about the contribution of the present work can be found in the introduction.}}
\author{Adrien Banse\thanks{Adrien Banse is supported by the French Community of Belgium in the framework of a FNRS/FRIA grant (\textit{Full address: }Euler building, Av. Georges Lemaître 4-6/L4.05.01, 1348 Louvain-la-Neuve, Belgium; \textit{Email address}: \texttt{adrien.banse@uclouvain.be}.).}
	\and Zheming Wang
	\and Raphaël M. Jungers\thanks{Raphaël Jungers is a FNRS honorary Research Associate. This project has received funding from the European Research Council (ERC) under the European Union's Horizon 2020 research and innovation programme under grant agreement No 864017 - L2C. Raphaël Jungers is also supported by the Innoviris Foundation and the FNRS (Chist-Era Druid-net). He is currently on sabbatical leave at Oxford University, Department of Computer Science, Oxford, UK.}
	\thanks{Adrien Banse and Raphaël M. Jungers are with ICTEAM, UCLouvain, and Zheming Wang is with the Zhejiang University of Technology.}
}
\usepackage{amsmath}
\usepackage{amsthm}
\usepackage{amssymb}
\usepackage{fullpage}
\usepackage{cite}
\usepackage{hyperref}

\usepackage{color}
\usepackage{tikz}
\usepackage{pdfpages}
\usepackage{svg}
\usepackage{lineno}
\usepackage{enumerate}

\usetikzlibrary{decorations.text}
\usetikzlibrary{automata,positioning}

\newtheorem{thm}{Theorem}

\newtheorem{prop}{Proposition}

\newtheorem{lemma}{Lemma}
\theoremstyle{definition}
\newtheorem{ass}{Assumption}
\newtheorem{defn}{Definition}
\newtheorem{prob}{Problem}
\newtheorem{rem}{Remark}

\newcommand*{\Scale}[2][4]{\scalebox{#1}{$#2$}}%

\let\citep\cite

\begin{document}
	\maketitle
	\begin{abstract}
		 We present a data-driven framework based on Lyapunov theory to provide stability guarantees for a family of hybrid systems. In particular, we are interested in the asymptotic stability of switching linear systems whose switching sequence is constrained by labeled graphs, namely constrained switching linear systems. In order to do so, we provide chance-constrained bounds on stability guarantees, that can be obtained from a finite number of noisy observations. 

We first present a method providing stability guarantees from sampled trajectories in the hybrid state-space of the system. We then study the harder situation where one only observes the continuous part of the hybrid states. We show that in this case, one may still obtain formal chance-constrained stability guarantees. For this latter result we provide a new upper bound of general interest, also for model-based stability analysis.
	\end{abstract}
	
	\input{content/1_introduction.tex}

	\input{content/2_preliminaries.tex}
	\input{content/3_hybrid_states.tex}

	\input{content/4_continuous_states.tex}

	\input{content/5_applications.tex}

	\input{content/6_conclusions.tex}
	\input{content/7_appendix.tex}

	\bibliographystyle{ieeetr}
	\bibliography{ref.bib}

\end{document}

%% file: content/1_introduction.tex
\section{Introduction}
Recently, industry, society and technology have increasingly relied on complex systems. For example, \emph{Cyber-Physical Systems} (or \emph{CPS} for short) are dynamical systems in which there are constant interactions between physical and digital, computational devices \citep{alur_2015_cps}. Such interactions lead to hybrid dynamics, where discrete and continuous dynamical behaviours are intertwined, placing the analysis of \emph{hybrid systems} at the center of a great research effort \citep{tabuda_2009_verification, goebel_2012_hybrid}. 

\emph{Switching systems} are dynamical systems in which a \emph{switching rule} plays a non-trivial role \citep{liberzon_1999_basic, liberzon_2003_switching, jungers_joint_2009}. These systems consist of an important family of hybrid systems, that often arise in CPS applications \citep{tabuda_2009_verification}. In particular, switching \emph{linear} systems have been extensively studied recently. Indeed, many model-based techniques have been developed in the past years to analyze their stability \citep{rota_1960_jsr, parrilo_2007_SOS, jungers_joint_2009, ahmadi_2014_pathcomplete}. Some of these tools have been extended to more general models, in which the switching rule is constrained by a labeled graph. These systems, on which we focus in this paper, are called \emph{constrained switching linear systems} (or \emph{CSLS} for short) and have an extensive range of applications, such as Networked Control Systems \citep{donkers_2011_NCS, jungers_2016_NCS}, or coordination
of multi-agents systems \citep{jadbabaie_2003_coordination}. In this context, the \emph{constrained joint spectral radius} (\emph{CJSR} for short), as well as new families of Lyapunov functions have been introduced to take into account these constraints \citep{dai_gelfand-type_2012, philippe2014converse, philippe_stability_2016}. 

In many practical applications, the engineer may not have access to the model of the dynamical system. This could occur for multiple reasons, such as the intrinsic complexity of the model, or for security purposes. In such cases, one cannot rely on model-based methods such as those cited above, but rather has to analyze the underlying system in a data-driven fashion. Moreover, the central tool to study stability of switching system, the joint spectral radius, is known to be NP-hard to compute, and without any algebraic solution \citep{blondel_2000_NP}. For these reasons, novel data-driven techniques have been recently developed for switching linear systems, but where the switching rule is supposed to be arbitrary. Some of them provide chance-constrained bounds on the joint spectral radius \citep{kenanian_data_2019, berger_chance-constrained_2021, rubbens2021datadriven}, using \emph{scenario optimization} as their main tool \citep{campi_2018_introduction}. Other methods extend classical tools from machine learning such as $k$-means clustering \citep{lauer_2013_regression}. Recently, \citep{lauer_2022_statistical} also provides formal guarantees on switching linear system identification based on statistical learning theory. In this paper, we present a method based on the scenario optimization to compute an approximate upper bound on the CJSR. Our approach aims at providing a data-driven method for analyzing an unknown system, but could also be used as a randomized analysis for (model-based) computation of the CJSR of a known system.

\paragraph{Contributions} In this paper, we provide a data-driven method to assess the quality of approximation of the CJSR of a CSLS based on a finite set of observations of trajectories, thereby generalizing \cite{kenanian_data_2019, berger_chance-constrained_2021, rubbens2021datadriven} to the constrained case and therefore taking into account hybrid state-spaces, which consists in a supplementary technical challenge. We investigate several different types of problems, depending on the information available on the trajectories. Preliminary versions of our results are presented in \cite{cdc22, mtns22}, sometimes without proofs, or using a different technique. On top of providing complete proofs our contribution is threefold compared with the latter references. First, we present a unified framework for the different types of bounds, using chance-constrained results on quasi-linear programs \cite{berger_chance-constrained_2021}. Second, we extend our methods to multiple steps and noisy observations, whereas 1-step observations are assumed in \cite{cdc22}, and noise-free observations are assumed in \cite{cdc22, mtns22}. In particular, our method allows to measure the impact of the noise on the approximation of the CJSR, which represents an additional motivation for our data-driven approach compared to the references cited above. Moreover, our results do not depend on the distribution of the noise, which is an asset for practical applications. Finally, we demonstrate our methods on a new numerical example, corresponding to a real-life application.

\paragraph{Outline} The rest of this paper is organized as follows. In Section~\ref{preli}, we first define model-based tools for studying stability of CSLS, and then we present chance-constrained optimization results. In Section~\ref{hybridStates}, we provide a probabilistic upper bound on the value of the CJSR, assuming access to a finite number of hybrid states. In Section~\ref{continuousStates}, we show that we are still able to provide a probabilistic upper bound on the CJSR assuming that we only have access to less information, that is only continuous states. Finally, results are illustrated in Section~\ref{application} on a numerical example corresponding to a real application, namely Network Control Systems.

\paragraph{Notations}
In this work, $\mathbb{R}$ is the set of real numbers, and $\mathbb{Z}$ the set of integers. In particular, $\mathbb{Z}_{> 0}$ denotes the strictly positive integers. $\mathbb{R}_{> 0}$ and $\mathbb{R}_{\geq 0}$ respectively denote the set of positive and non-negative real numbers. For the sake of readability, for any $m \in \mathbb{Z}_{> 0}$ we define $[m] = \{1, \dots, m\}$. The set of real vectors of dimension $n$ is noted $\mathbb{R}^n$. The set of real matrices of dimension $n\times n$ is noted $\mathbb{R}^{n \times n}$, and the set of symmetric positive definite matrices is denoted by $\mathcal{S}^n$. For any real matrix $A \in \mathbb{R}^{n \times n}$, $\lambda^{\max}(A)$ and $\lambda^{\min}(A)$ are respectively the maximal and minimal eigenvalues of $A$, and $\det(A)$ is its determinant. In terms of norms, $\|\cdot\|$ denotes the Euclidian norm, and $\|\cdot\|_F$ the Frobenius norm. The set $\mathbb{S}^n$ is the unit sphere of dimension $n$, that is $\{x \in \mathbb{R}^n \,|\, \|x\| = 1 \}$, $\mathbb{B}^n$ is the unit ball of dimension $n$, that is $\{x \in \mathbb{R}^n \,|\, \|x\| \leq 1 \}$, and, for any $W \in \mathbb{R}_{>0}$, $B^n[W]$ is the ball of radius $W$, that is $\{x \in \mathbb{R}^n \,|\, \|x\| \leq W \}$. Finally, for $\Omega \subset \mathbb{R}^n$ and $A \in \mathbb{R}^{n \times n}$, the image of $\Omega$ by applying $A$ is noted $A\Omega = \{Ax \,|\, x \in \Omega \}$, the projection of $\Omega$ on $\mathbb{S}^n$ is noted $\Pi_{\mathbb{S}^n}(\Omega)$ and $\text{conv}(\Omega)$ stands for the convex hull of $\Omega$.

%% file: content/2_preliminaries.tex
\section{Preliminaries} \label{preli}

In this section, we recall some concepts needed to present our methods. We first provide a formal definition of CSLS. Then, we recall some principles of data-driven analysis of these systems from \citep{dai_gelfand-type_2012, philippe_stability_2016}. In particular, we define the notion of asymptotic stability, the property that we investigate for such dynamical systems, as well as a Lyapunov framework to study it. 

First, we formally define the labeled graph.
\begin{defn}[Labeled graph, {\citep[Definition~2.2.1~and~3.1.1]{lind_introduction_1995}}] A \emph{graph} $G$ is composed of a finite set of \emph{nodes} $\mathcal{U}(G)$ and a finite set of \emph{edges} $\mathcal{E}(G)$. Each edge $e \in \mathcal{E}(G)$ \emph{starts} at $s(e) \in \mathcal{U}(G)$, and \emph{terminates} at $t(e) \in \mathcal{U}(G)$. For a given $m \in \mathbb{Z}_{>0}$, a \emph{labeled graph} $\mathcal{G}$ is a pair $(G, \sigma)$ where $G$ is a graph and $\sigma : \mathcal{E}(G) \to [m]$ is the \emph{labeling} and assigns to each edge $e \in \mathcal{E}(G)$ a label $\sigma(e)$ in the finite set $[m]$.
\end{defn}
Note that for a given labeled graph $\mathcal{G} = (G, \sigma)$, $\mathcal{U}(G)$ and $\mathcal{E}(G)$ are respectively noted $\mathcal{U}$ and $\mathcal{E}$ when it is clear from the context. Starting from a labeled graph $\mathcal{G} = (G, \sigma)$, and for $l \in \mathbb{Z}_{> 0}$, the \emph{language of $\mathcal{G}$ restricted to length $l$}, noted $\mathcal{L}_{\mathcal{G}, l}$ is the set of all possible sequences of labels of length $l$ in the labeled graph $\mathcal{G}$. It is defined as 
\begin{equation} 
	\mathcal{L}_{\mathcal{G}, l} = \left\{ (\sigma(e_0), \dots, \sigma(e_{l-1})) \, | \, \forall e_0 \in \mathcal{E}, i \in [l-2] : t(e_i) = s(e_{i+1}) \right\}.
\end{equation}
The \emph{language} of the labeled graph $\mathcal{G}$, noted $\mathcal{L}_{\mathcal{G}}$, is the set of all possible sequences in the labeled graph $\mathcal{G}$, that is
$
	\mathcal{L}_{\mathcal{G}} = \bigcup_{l \in \mathbb{Z}_{> 0}}^{\infty} \mathcal{L}_{\mathcal{G}, l}.
$
An illustration of a labeled graph, is given in Figure~\ref{labeledgraphsimple}. 

\begin{figure}[h!]
\centering
\begin{tikzpicture}[shorten >=1pt,node distance=4cm,on grid,auto] 
   \node[state] (a) {$a$}; 
   \node[state] (b) [right = of a] {$b$}; 
   \path[->]
   	(a) edge[bend left] node {2} (b)
	(b) edge[bend left] node {1} (a)
	(a) edge[loop left] node {1} (a);
\end{tikzpicture}
\caption{Illustration of a labeled graph $\mathcal{G} = (G,\sigma)$, where $\mathcal{U}(G) = \{a, b\}$, and where $\mathcal{E}(G) = \{e_1, e_2, e_3\}$ such that $s(e_1) = a$, $t(e_1) = a$, $\sigma(e_1) = 1$, $s(e_2) = a$, $t(e_2) = b$, $\sigma(e_2) = 2$, $s(e_3) = b$, $t(e_3) = a$, $\sigma(e_3) = 1$. The language $\mathcal{L}_{\mathcal{G}}$ contains all the words, excepted those that have two consecutive occurences of "$2$" in them.}
\label{labeledgraphsimple}
\end{figure}
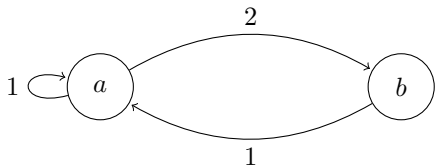

In this work, we are interested in a particular subtype of discrete-time hybrid systems, namely constrained switching systems. A hybrid system is a dynamical system whose dynamics is characterized by both continuous and discrete behaviors \cite{goebel_2012_hybrid, alur_2015_cps}. Starting from $m \in \mathbb{Z}_{> 0}$, consider a labeled graph $\mathcal{G} = (G, \sigma)$ where we recall that the labeling is defined as $\sigma : \mathcal{E}(G) \to [m]$. Consider also a finite set of matrices $\mathcal{A} = \{A_1, \dots, A_m\}$. A \emph{constrained switching linear system} (or \emph{CSLS} for short) is defined by states of the form $(e, x)$ taking values in the hybrid \emph{state-space} $\mathcal{E}(G) \times \mathbb{R}^n$. The dynamics of a CSLS is given by the following hybrid relations: 
\begin{equation}
\label{hybrid_dynamics}
\begin{aligned}
	e(k+1) &\in \{e \in \mathcal{E}(G) \, | \, s(e) = t(e(k)) \} &\textrm{(discrete behavior)}, \\
	x(k+1) &=A_{\sigma(e(k+1))}x(k)&\textrm{(continuous behavior)}, 
\end{aligned}
\end{equation}
with $e(0) = e_0 \in \mathcal{E}$ and $x(0) = x_0 \in \mathbb{R}^n$.

A CSLS is thus entirely characterized by a labeled graph $\mathcal{G} = (G, \sigma)$, and a finite set of matrices $\mathcal{A}$. It is therefore naturally noted $S(\mathcal{G}, \mathcal{A})$.

\begin{rem} \label{flower}
	An \emph{arbitrary switching linear system} (\emph{ASLS} for short), i.e. a switching system whose switching rule is unconstrained, can be seen as a special case of a CSLS. An ASLS defined by a set of matrices $\mathcal{A} = \{A_1, \dots, A_m\}$ is a CSLS $S(\mathcal{F}_m, \mathcal{A})$ where $\mathcal{F}_m = (F_m, \sigma)$ with $F_m$ a directed graph only composed of one node and $m$ self-loops $e_1, \dots, e_m$, and $\sigma$ such that, for $i = 1, \dots, m$, $\sigma(e_i) = i$. An illustration of $\mathcal{F}_m$ is provided in Figure~\ref{flowerASLS}.
\end{rem}
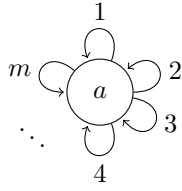
\begin{figure}[h!] \centering
\begin{tikzpicture}[shorten >=1pt,node distance=3cm,on grid,auto] 
   \node[state] (1) {$a$}; 
   \path[->]
   (1) edge [out=70,in=110,looseness=5] node [above] {1} (1)
   (1) edge [out=0,in=40,looseness=5] node [right] {2} (1)
   (1) edge [out=-70,in=-110,looseness=5] node [below] {4} (1)
   (1) edge [out=140,in=180,looseness=5] node [left] {$m$} (1)
   (1) edge [out=-10,in=-50,looseness=5] node [right] {3} (1);
   \node[] at (-0.9,-0.5) {$\ddots$};
\end{tikzpicture}
\caption{Illustration of the labeled graph $\mathcal{F}_m = (F_m, \sigma)$. The language $\mathcal{L}_{\mathcal{F}_m}$ contains all words composed of $1, 2, \dots, m$.} 
\label{flowerASLS}
\end{figure}

In this work, we are interested in the asymptotic stability of such hybrid systems. A hybrid system is \emph{asymptotically stable} (or stable, for simplicity) if, for any $(e_0, x_0) \in \mathcal{E} \times \mathbb{R}$ such that $e(0) = e_0$ and $x(0) = x_0$, 
\begin{equation}
	\lim_{k \to \infty} \|x(k)\| = 0.
\end{equation}

To study stability, the \emph{constrained joint spectral radius} (\emph{CJSR} for short), has been introduced in \cite{dai_gelfand-type_2012} to generalize the notion of \emph{joint spectral radius} \cite{jungers_joint_2009} to the constrained case.
\begin{defn}[Constrained joint spectral radius, \citep{dai_gelfand-type_2012}] \label{CJSRdef}
	Given a labeled graph $\mathcal{G}$ and a set of matrices $\mathcal{A}$, the \emph{constrained joint spectral radius} of the CSLS $S(\mathcal{G}, \mathcal{A})$, noted $\rho(\mathcal{G}, \mathcal{A})$, is defined as 
	\begin{equation}
		\rho(\mathcal{G}, \mathcal{A}) 
		= \lim_{k \to \infty} \max_{(\sigma(0), \dots, \sigma(t-1)) \in \mathcal{L}_{\mathcal{G}, k}} \left\| A_{\sigma(k-1)} \dots A_{\sigma(0)} \right\|^{1/k}.
	\end{equation}
\end{defn} 

As the following proposition shows, the CJSR characterizes the stability of a CSLS.
\begin{prop}[{\citep[Corollary~2.8]{dai_gelfand-type_2012}}] \label{CJSRstability}
	Given a labeled graph $\mathcal{G}$ and a set of matrices $\mathcal{A}$, the CSLS $S(\mathcal{G}, \mathcal{A})$ is stable if and only if $\rho(\mathcal{G}, \mathcal{A}) < 1$.
\end{prop}

Using a slightly different formalism than in \citep{philippe_stability_2016}\footnote{Definition~\ref{lyap_def} is equivalent to the definition of \emph{multinorms} in \citep{philippe_stability_2016}. We prefer the more general term \emph{Lyapunov function}, that we define in the hybrid framework in Definition~\ref{lyap_def}, and that we generalize by a slight abuse of language to values of $\gamma$ that may be greater or equal to 1.}, for some $\gamma \in \mathbb{R}_{>0}$, we now introduce the notion of Lyapunov function for CSLS.
\begin{defn}[Lyapunov function]
	\label{lyap_def}
	Given a labeled graph $\mathcal{G} = (G, \sigma)$ and a finite set of matrices $\mathcal{A} = \{A_1, \dots, A_m\}$, a \emph{Lyapunov function} for a CSLS $S(\mathcal{G}, \mathcal{A})$ is a function $V : \mathcal{U} \times \mathbb{R}^n \to \mathbb{R}_{\geq 0}$ such that, for all $u \in \mathcal{U}$, the function $V(u, \cdot)$ is a \emph{norm}\footnote{In the sense that it is \emph{subadditive}, \emph{absolutely homogeneous} and \emph{positive definite}.}, and for all $x \in \mathbb{R}^n$, $e \in \mathcal{E}$, there exists $\gamma \in \mathbb{R}_{> 0}$ such that
		\begin{equation}
			V\left(t(e), A_{\sigma(e)}x\right) \leq \gamma V(s(e), x).
		\end{equation}
	The quantity $\gamma(V)$ associated to the Lyapunov function $V$ is defined as 
	\begin{equation}
		\gamma(V) = \min_\gamma \{ \gamma \, |\,  \forall e \in \mathcal{E}, x \in \mathbb{R}^n :  V\left(t(e), A_{\sigma(e)}x\right) \leq \gamma V(s(e), x) \}.
	\end{equation}
\end{defn}
It has been proven in \citep[Proposition~2.2]{philippe_stability_2016} that the CJSR of a given CSLS is also defined as follows:
\begin{equation} \label{CJSRotherdef}
	\rho(\mathcal{G}, \mathcal{A}) = \inf_V \{ \gamma(V) \,|\, V \text{ is a Lyapunov function for } S(\mathcal{G}, \mathcal{A})\}.
\end{equation}

In this work, we restrict ourselves to the family of quadratic Lyapunov functions.
\begin{defn}[Quadratic Lyapunov function]
	Given a labeled graph $\mathcal{G} = (G, \sigma)$ and a finite set of matrices $\mathcal{A} = \{A_1, \dots, A_m\}$, a \emph{quadratic Lyapunov function} is a Lyapunov function for the CSLS such that, for all $u \in \mathcal{U}$, $x \in \mathbb{R}^n$, $V(x, u) = \| x \|_{P_u} = \sqrt{x^T P_u x}$ for some $P_u \in \mathcal{S}^n$. 
\end{defn}

Let us now define the following quantity:
\begin{equation} \label{quadraticupperboundCJSR}
	\gamma^*(\mathcal{G}, \mathcal{A}) = \inf_V \{ \gamma(V) \,|\, V \text{ is a quadratic Lyapunov function for } S(\mathcal{G}, \mathcal{A})\}.
\end{equation}
Since the family of quadratic Lyapunov functions is a subset of the family of all Lyapunov functions, the latter quantity is an upper bound on the CJSR \citep[Theorem~3.1]{philippe_stability_2016}: 
\begin{equation}
	\rho(\mathcal{G}, \mathcal{A}) \leq \gamma^*(\mathcal{G}, \mathcal{A}).
\end{equation}

We now introduce the notion of \emph{$l$-lifting} of a CSLS. This notion and its link to the CJSR have been introduced in \citep[Section 3.1]{philippe_stability_2016}. 
\begin{defn}[$l$-product lifting, \cite{philippe_stability_2016}]
\label{product_lifting}
	Given a labeled graph $\mathcal{G} = (G, \sigma)$, a finite set of matrices $\mathcal{A} = \{A_1, \dots, A_m\}$, and a length $l \in \mathbb{N}_{> 0}$, the \emph{$l$-product lift} of $S(\mathcal{G}, \mathcal{A})$ is a CSLS $S(\mathcal{G}^{(l)}, \mathcal{A}^{(l)})$, with 
	\begin{itemize}
		\item $\mathcal{G}^{(l)} = (G^{(l)}, \sigma^{(l)})$, where $\mathcal{U}(G^{(l)}) = \mathcal{U}(G)$, and $e' \in \mathcal{E}(G^{(l)})$ if there exists a path $(e_0, \dots, e_{l-1})$ in $G$ with $s(e') = e_0$ and $t(e') = e_{l-1}$, in which case 
		\begin{equation}
			\sigma^{(l)}(e') = (\sigma(e_0), \dots, \sigma(e_{l-1})) \in \mathcal{L}_{\mathcal{G}, l}.
		\end{equation}
		\item $\mathcal{A}_{\mathcal{G}}^{(l)} = \{ A_{\sigma(e_{l-1})} \dots A_{\mathcal{L}(e_0)} \,|\, (\sigma(e_0), \dots, \sigma(e_{l-1}) \in \mathcal{L}_{\mathcal{G}, l} \}$. For any $e' \in \mathcal{E}(G^{(l)})$ corresponding to the initial path $(e_{0}, \dots, e_{l-1})$, the matrix $A_{\sigma^{(l)}(e')} = A_{\sigma(e_{l-1})}\dots A_{\sigma(e_{0})}$.
	\end{itemize}
\end{defn}

The $l$-lifting $S(\mathcal{G}^{(l)}, \mathcal{A}_{\mathcal{G}}^{(l)})$ of any CSLS $S(\mathcal{G}, \mathcal{A})$ must be interpreted in the following way: for some $l, N \in \mathbb{N}_{>0}$, if $(x_0, e_0), (x_1, e_1), (x_2, e_2), \dots, (x_{lN}, e_{lN})$ is a trajectory of $S(\mathcal{G}, \mathcal{A})$, then $(x_0, e_0), (x_{l}, e_{l}), (x_{2l}, e_{2l}),  \dots,$ $(x_{lN}, e_{lN})$ is a trajectory of $S(\mathcal{G}^{(l)}, \mathcal{A}_{\mathcal{G}}^{(l)})$. An illustration of the labeled graph $\mathcal{G}^{(l)}$ is given in Figure~\ref{labeledgraphlifted}.

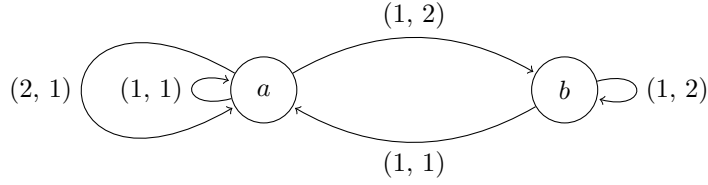
\begin{figure}[h]
\centering
\begin{tikzpicture}[shorten >=1pt,node distance=4cm,on grid,auto] 
   \node[state] (a) {$a$}; 
   \node[state] (b) [right = of a] {$b$} ; 
   \path[->]
   	(a) edge[bend left] node {(1, 2)} (b)
	(b) edge[bend left] node {(1, 1)} (a)
	(b) edge[loop right] node {(1, 2)} (b)
	(a) edge[loop left] node {(1, 1)} (a)
	(a) edge[loop left, in = 210, out = 150, looseness = 18] node {(2, 1)} (a);
\end{tikzpicture}
\caption{Illustration of $\mathcal{G}^{(2)}$, the $2$-lifting labeled graph of $\mathcal{G}$, illustrated in Figure~\ref{labeledgraphsimple}.}
\label{labeledgraphlifted}
\end{figure}

Using this lift for $l > 1$ allows to refine the upper bound on the CJSR obtained with quadratic Lyapunov functions, compared to the classical case $l = 1$. Indeed, it has been proven in \citep[Theorem~3.2]{philippe_stability_2016} that 
\begin{equation} \label{l-lifting-eq}
	\rho(\mathcal{G}, \mathcal{A}) \leq \gamma^*(\mathcal{G}^{(l)}, \mathcal{A}_{\mathcal{G}}^{(l)})^{\frac{1}{l}} \leq \rho(\mathcal{G}, \mathcal{A})^{\frac{1}{nl}}, 
\end{equation}
where $n$ is the dimension of the system.

In the following sections, we will show how to construct a quadratic Lyapunov function from a finite amount of data. Using the notion of $l$-lifting, we will present a general method assuming $l$-steps observations. This will then give us an approximate upper bound on the value of the CJSR of the unknown CSLS, extending geometrical properties used in \cite{kenanian_data_2019}.

%% file: content/3_hybrid_states.tex
\section{Learning stability guarantees from observations of hybrid states}
\label{hybridStates}

In this section, we consider a CSLS $S(\mathcal{G}, \mathcal{A})$, where the labeled graph $\mathcal{G}$ and the set of matrices $\mathcal{A}$ are supposed to be unknown. Similarly to \citep{kenanian_data_2019, berger_chance-constrained_2021} in the unconstrained case, we will approximate the CJSR of this system from a finite number of observations of the hybrid state-space of the CSLS. Finally, this will allow for deriving probabilistic stability guarantees for the unknown system.

\subsection{Problem formulation} \label{hybrid_pf}

For some unknown labeled graph $\mathcal{G} = (G, \sigma)$ and set of matrices $\mathcal{A} = \{A_1, \dots, A_m\}$, consider the CSLS $S(\mathcal{G}, \mathcal{A})$, for which one can only sample a finite set of $N$ observations of length $l$. Consider its $l$-product lift $S(\mathcal{G}^{(l)}, \mathcal{A}_{\mathcal{G}}^{(l)})$, where $\mathcal{G}^{(l)} = (G^{(l)}, \sigma^{(l)})$ according to Definition~\ref{product_lifting}. A \emph{sample of length $l$} is defined as the triplet $(x, e, w)$, where 
\begin{itemize}
	\item $x \in \mathbb{S}^n$ is an initial continuous state: $x(0) = x$ in \eqref{hybrid_dynamics}\footnote{We can restrict ourselves to the unit sphere without loss of generality thanks to \emph{homogeneity} of the dynamical system: for all $\mu > 0$ and $A \in \mathcal{A}$, $A(\mu x) = \mu Ax$.}. It is sampled according to the uniform probability measure $\mu^{x}$ on $\mathbb{S}^n$.
	\item $e \in \mathcal{E}(G^{(l)})$ is an edge of the $l$-product lift. It is assumed to be sampled according to the uniform probability measure $\mu_{l}^{e}$ on $\mathcal{E}(G^{(l)})$.
	\item $w \in B^n[W]$ is a noise on the continuous state-space, where $W > 0$ is a known upper bound. It is sampled according to some probability measure $\mu^w$ on $B^n[W]$.
\end{itemize}
The \emph{sample set} is therefore defined as 
\begin{equation} \label{hybrid_observations}
	\omega_N = \left\{\left(x_i, e_i, w_i\right)\right\}_{i = 1, \dots, N} \subset \Delta, 
\end{equation}
where $\Delta = \mathbb{S}^n \times \mathcal{E}(G^{(l)}) \times B^n[W]$. Each element $(x_i, e_i, w_i) \in \omega_N$ is sampled independently of the others and according to the product probability measure $\pi_l = \mu^x \times \mu^e_l \times \mu^w$.

\begin{rem} \label{hybrid_link_samples_observations}
In the context of this section, we define an \emph{observation} as a hybrid state $(x, u) \in \mathbb{R}^n \times \mathcal{U}(G)$, that is a couple of continuous state and node. The link between the sample $(x_i, e_i, w_i)$ and the pair of observations $(x_i, u_i), (y_i, v_i)$, is that $y_i = A_{\sigma^{(l)}(e_i)}x_i + w_i$, where $s(e_i) = u_i$, and $t(e_i) = v_i$. It is important to highlight that, even though the experiments are defined by the samples $(x_i, e_i, w_i)$, one only observes the pairs of hybrid states $(x_i, u_i), (y_i, v_i)$ as described above.
\end{rem}

Our main problem is stated as follows.
\begin{prob} \label{problem1}
	Consider an unknown CSLS $S(\mathcal{G}, \mathcal{A})$, where $\mathcal{G} = (G, \sigma)$. For a fixed length $l \in \mathbb{Z}_{>0}$, assume that one is given $N$ pairs of observations, as described in Remark~\ref{hybrid_link_samples_observations}. For some user-defined confidence level of at least $\beta \in (0, 1)$, and, knowing $|\mathcal{U}(G)|$, the number of nodes, as well as $|\mathcal{E}(G^{(l)})|$, the number of edges of the $l$-lifting, provide an upper bound $ \rho(\mathcal{G}, \mathcal{A}) \leq \overline{\rho}$.
\end{prob}
\begin{rem}
	Although the uniform setting was chosen in the context of this work, we attract the attention of the reader on the fact that more general settings can be considered. For example, one can assume that the distribution on the edges is unknown, and that only an upper bound on the number of nodes and edges is known \citep{mtns22}. One can also assume that labels are sampled instead of edges \citep{cdc22}. However, because the latter introduce conservativism, we choose not to present the last-mentioned frameworks. This remark holds for the results of Section~\ref{continuousStates} as well.
\end{rem}

\subsection{Results}

In this section, we propose a solution for Problem~\ref{problem1}. In order to do so, following the lifting result \eqref{l-lifting-eq}, we will approximate a quadratic Lyapunov function for the unknown $l$-product lift $S\left(\mathcal{G}^{(l)}, \mathcal{A}^{(l)}\right)$. Let us thus introduce the following quadratic \emph{sampled program}\footnote{In equation \eqref{scenario_program_1}, $(\lambda, c(P_1, \dots, P_{\mathcal{U}}))$ means that we optimize the objective in the \emph{lexicographical order}. That is, $(\lambda_1, c_1) > (\lambda_2, c_2)$ if $\lambda_1 > \lambda_2$, or $\lambda_1 = \lambda_2$ and $c_1 > c_2$, and $(\lambda_1, c_1) = (\lambda_2, c_2)$ if $\lambda_1 = \lambda_2$ and $c_1 = c_2$.}:
\begin{subequations}
\label{scenario_program_1}
\begin{align}
	\mathcal{P}(\omega_N): &\min_{\substack{\lambda \in \mathbb{R}_{\geq 0} \\ \forall u \in \mathcal{U}: P_u \in \mathbb{R}^{n \times n}}} \left(\lambda, c\left(P_1, \dots, P_{|\mathcal{U}|}\right)\right) \\
	\text{s.t. } 
	&\forall u \in \mathcal{U}:\,  P_u \in \mathcal{X} = \{ P \in \mathbb{R}^{n \times n} \, | \, P \succeq I, \| P \|_F \leq C \}, \label{PgeqI}  \\
	&\forall (x, e, w) \in \omega_N: 
	\left\| 
	A_{\sigma^{(l)}(e)} x + w
	\right\|_{P_{t\left(e\right)}}
	\leq 
	\lambda^{l} 
	\left\|
	x
	\right\|_{P_{s\left(e\right)}} \label{lyap_constraint}, 
\end{align}
\end{subequations}
where $c: \mathcal{X^{|\mathcal{U}|}} \to \mathbb{R} : (P_1, \dots, P_{|\mathcal{U}|}) \mapsto \max (\|P_1\|_F, \dots, \|P_{|\mathcal{U}|}\|_F)$, and $C$ is a large predefined value (say, $C = 10^6$).

We denote the optimal solutions of the sampled program~\eqref{scenario_program_1} by $\lambda^*(\omega_N)$ and $\left\{ P_u^*(\omega_N)\right\}_{u \in \mathcal{U}}$. This program is a data-driven version of the approximation of the CJSR with a quadratic Lyapunov function, as described in~\eqref{quadraticupperboundCJSR}, and using $l$-product liftings, as shown in Equation~\eqref{l-lifting-eq}. In addition to generalizing data-driven programs expressed in \citep{kenanian_data_2019, berger_chance-constrained_2021} to the constrained case, our formulation takes into account noise, as one can see in constraints~\eqref{lyap_constraint}.

We now define the concept of \emph{non-degenerate sample set} for a sampled program.
\begin{defn}[Non-degeneracy, {\citep[Definition 2.7]{calafiore_2010_random}}] \label{non_degenerate_property}
	The sample set $\omega_N = \{\delta_1, \dots, \delta_N\} \subset \Delta$ is \emph{non-degenerate} if there exists a unique set $J \subset [N]$ for which $\beta = \{ \delta_j \}_{j \in J}$ is such that $(\lambda^*(\beta), c^*(\beta)) = (\lambda^*(\omega_N), c^*(\omega_N))$.
\end{defn}

Given the following assumption, we show that all samples sets considered in this work are non-degenerate.

\begin{ass} \label{ass_on_w_dis}
	All sets $S \subseteq B^n[W]$ such that $\mu^w(S) > 0$ have non-zero Lebesgue measure.
\end{ass}

\begin{prop} \label{non_degenerate_for_sure}
	Let $\omega_N$ be a sample set as described in \eqref{hybrid_observations}. Let Assumption~\ref{ass_on_w_dis} hold, then $\omega_N$ is non-degenerate with probability 1.
\end{prop}
 A proof of Proposition~\ref{non_degenerate_for_sure} is provided in \ref{app}. For now, we consider that Assumption~\ref{ass_on_w_dis} holds, which is a reasonable assumption on the sampling, and therefore that the sample set is always almost-surely non-degenerate. Now, using \citep[Theorem~6]{berger_chance-constrained_2021}, we provide a chance-constrained result for the CJSR approximation problem in the presence of noise.


\begin{prop}
\label{PAC1}
	Consider an unknown CSLS $S(\mathcal{G}, \mathcal{A})$. For any given number of samples $N \in \mathbb{Z}_{>0}$ such that $N \geq d=  |\mathcal{U}|n(n+1)/2$, and for $l \in \mathbb{N}_{> 0}$, let $\omega_N \subset \Delta$ be the sample set as described in \eqref{hybrid_observations}. Given $\lambda^*(\omega_N)$ and $\left\{ P_u^*(\omega_N)\right\}_{u \in \mathcal{U}}$ the solutions of the sampled program~\eqref{scenario_program_1}, let 
	\begin{equation} \label{viol}
	V(\omega_N)
		= \left\{ (x, e, w) \in \Delta \, | \,
		\left\|A_{\sigma^{(l)}(e)} x
		\right\|_{P_{t\left(e\right)}^*(\omega_N)}
		>
		\overline{\lambda}_e(\omega_N)
		\left\|x\right\|_{P_{s\left(e\right)}^*(\omega_N)}
	\right\}, 
\end{equation}
where 
	\begin{equation} \label{lambda_with_noise}
		\overline{\lambda}_e(\omega_N) = \lambda^*(\omega_N)^l + \sqrt{\frac{\lambda^{\max}\left(P^*_{t(e)}(\omega_N)\right)}{\lambda^{\min}\left(P^*_{s(e)}(\omega_N)\right)}}W.
	\end{equation}
	Then, for any confidence level $\beta \in (0, 1)$, the following holds: 
	\begin{equation}
		\pi_l^N\left(
			\left\{
				\omega_N \subset \Delta \, | \, \pi_l\left(
					V(\omega_N)
				\right) \leq \varepsilon(\beta, N)
			\right\}
		\right) \geq 1 - \beta, 
	\end{equation} 
	with 
	\begin{equation} \label{vareps}
		\varepsilon(\beta, N) = \Phi^{-1}(1 - \beta; d-1, N).
	\end{equation} 
	where $\Phi(\cdot; d-1, N)$ is the \emph{regularized incomplete beta function}\footnote{See e.g. \citep[Equation~(6.6.2)]{crowder_2011_NIST}.}
\end{prop}
\begin{proof}
	For the sake of readability, let $\lambda = \lambda^*(\omega_N)$, and, for $e \in \mathcal{E}$, let $A_e = A_{\sigma^{(l)}(e)}$, $\overline{\lambda}_e = \overline{\lambda}_e(\omega_N)$, $P_{t(e)} = P_{t(e)}^*(\omega_N)$ and $P_{s(e)} = P_{s(e)}^*(\omega_N)$.
	\begin{enumerate}[(i)]
		\item \label{prop_proof_item1} 
		One can see that program~\eqref{scenario_program_1} is the \emph{sampled form} of a \emph{quasi-linear program}, as defined in \citep{berger_chance-constrained_2021}. Indeed, $c$ is strongly convex, the admissible set for the variables is $\mathcal{X}^{|\mathcal{U}|}$, which is compact, and $\omega_N$ is a finite subset of $\Delta$. Now, the set
		\begin{equation}
		\tilde{V}(\omega_N) = \left\{ (x, e, w) \in \Delta \, | \,
		\left\|A_e x + w
		\right\|_{P_{t\left(e\right)}}
		>
		\lambda^{l}
		\left\|x\right\|_{P_{s\left(e\right)}}
		\right\}
		\end{equation}
		is the set of violated constraints in the quasi-linear program. Following \citep[Theorem~6]{berger_chance-constrained_2021}, since $N \geq d$, and that $\omega_N$ is almost surely non-degenerate following Proposition~\ref{non_degenerate_for_sure}, then 
		\begin{equation}
		\pi_l^N\left(
			\left\{
				\omega_N \subset \Delta \, | \, \pi_l\left(
					\tilde{V}(\omega_N)
				\right) \leq \varepsilon(\beta, N)
			\right\}
		\right) \geq 1 - \beta
	\end{equation} 
		holds.
	
		\item \label{prop_proof_item2}
		Second, we prove that $V(\omega_N) \subseteq \tilde{V}(\omega_N)$. 
		Let $(x, e, w) \in V(\omega_N)$, then	
		\begin{equation}
			\overline{\lambda}_e \| x \|_{P_{s(e)}} < \| A_e x\|_{P_{t(e)}}.
	\end{equation}
	By subadditivity of the norm, it holds that 
	\begin{align}
		\overline{\lambda}_e \|x\|_{P_{s(e)}} &< \| A_e x + w - w\|_{P_{t(e)}} \\
		\implies \overline{\lambda}_e \|x\|_{P_{s(e)}} &< \|A_e x + w\|_{P_{t(e)}} + \|w\|_{P_{t(e)}}.
	\end{align}
	Now, it is a well known fact that, for any vector $v \in \mathbb{R}^n$ and matrix $P \in \mathcal{S}^n$, $\sqrt{\lambda^{\min}(P)}\|v\| \leq \|v\|_P \leq \sqrt{\lambda^{\max}(P)}\|v\|$. Therefore, 
	\begin{equation}
	\begin{aligned}
		\lambda^l \|x\|_{P_{s(e)}} + & W \sqrt{\frac{\lambda^{\max}\left(P_{t(e)}\right)}{\lambda^{\min}\left(P_{s(e)}\right)}}  \sqrt{\lambda^{\min}\left(P_{s(e)}\right)} \\
		&< \|A_e x + w\|_{P_{t(e)}} + \sqrt{\lambda^{\max}\left( P_{t(e)} \right)}\|w\| \\
		&< \|A_e x + w\|_{P_{t(e)}} + \sqrt{\lambda^{\max}\left( P_{t(e)} \right)}W, 
	\end{aligned}
	\end{equation}
	which implies that
	\begin{equation}
		\lambda^l \|x\|_{P_{s(e)}} < \|A_e x + w\|_{P_{t(e)}}.
	\end{equation}
	\end{enumerate} 
	
	Now, \ref{prop_proof_item2} implies that $\pi_l(V(\omega_N)) \leq \pi_l(\tilde{V}(\omega_N))$. Therefore, for any $\varepsilon \in (0, 1)$, $\pi_l(\tilde{V}(\omega_N)) \leq \varepsilon$ implies that $\pi_l(V(\omega_N)) \leq \varepsilon$, and, following \ref{prop_proof_item1}, the proof is completed.
\end{proof}

From there, using geometrical properties of ellipsoids, one can use the measure of violated constraints of the program~\eqref{scenario_program_1} to derive guarantees on the model-based optimal objective value, that is $\rho(\mathcal{G}, \mathcal{A})$.

\begin{thm} \label{maintheorem1}
	Consider an unknown CSLS $S(\mathcal{G}, \mathcal{A})$. For any given number of observations $N \in \mathbb{Z}_{>0}$ such that $N \geq d=  |\mathcal{U}|n(n+1)/2$, and for $l \in \mathbb{N}_{> 0}$, let $\omega_N \subset \Delta$ be the sampled set as described in \eqref{hybrid_observations}. Given $\lambda^*(\omega_N)$ and $\{P_u^*(\omega_N)\}_{u \in \mathcal{U}}$ the solutions of the sampled program~\eqref{scenario_program_1}, for any confidence level $\beta \in (0, 1)$, the following holds:
	\begin{equation}
		\pi_l^N\left(
			\left\{
				\omega_N \subset \Delta \, 
				| \, \rho(\mathcal{G}, \mathcal{A}) \leq \overline{\rho}(\omega_N)
			\right\}
		\right) \geq 1 - \beta, 
	\end{equation} 
	where
	\begin{equation} \label{maintheorem1-bound}
		\overline{\rho}(\omega_N) = \max_{
		e \in \{e' \,|\, (x, e', w) \in \omega_N \}
		} \left( 
			\frac{\overline{\lambda}_e(\omega_N)}{\delta\left( \varepsilon(\beta, N) |\mathcal{E}(G^{(l)})| \kappa(P^*_{s(e)}(\omega_N)) / 2 \right)}
		\right)^{1/l}, 
	\end{equation}
	with $\varepsilon(\beta, N)$ and $\overline{\lambda}_e(\omega_N)$ are respectively defined in \eqref{vareps} and \eqref{lambda_with_noise}, 
	\begin{equation} \label{kappa}
		\kappa(P) = \sqrt{\det(P) / \lambda^{\min}(P)^n}, 
	\end{equation} 
	and, for any $x \in (0, 1)$, $\delta(x)$ is defined as  
	\begin{equation} \label{delta}
		\delta(x) = \sqrt{1 - \Phi^{-1}(2x; (n - 1) / 2, 1/2)}.
	\end{equation}
\end{thm}
\begin{proof}
	This theorem is a generalization of \citep[Corollary~11, Theorem~15]{kenanian_data_2019} to the constrained case, and with noisy data. For the sake of readibility, for any $e \in \mathcal{E}$, we write $A_e = A_{\sigma^{(l)}(e)}$, $P_{s(e)} = P^*_{s(e)}(\omega_N)$, $P_{t(e)} =  P^*_{t(e)}(\omega_N)$ and $\overline{\lambda}_e(\omega_N) = \overline{\lambda}_e$. 
	\begin{enumerate}[(i)]
		\item \label{item0} First, we show that, with probability at least $1-\beta$, there exists a set $\Omega \subset \mathbb{S}^n$ such that 
		\begin{equation} \label{eqwithomegaquad}
			\forall e \in \mathcal{E}, \forall x \in \mathbb{S}^n \setminus \Omega : 
		x^TA_eP_{t(e)}A_ex \leq \overline{\lambda}_e^2 x^TP_{s(e)}x, 
		\end{equation}
		and that satisfies
		\begin{equation}  \label{measure_omega}
			\mu^x(\Omega) \leq \varepsilon(\beta, N) |\mathcal{E}|.
		\end{equation}
		
		Condition \eqref{eqwithomegaquad} implies that
		\begin{equation} 
			\Omega = \{
				x \in \mathbb{S}^n \,|\, \exists e \in \mathcal{E} :
				\| A_e x \|_{P_{t(e)}}
				>
				\overline{\lambda}_e
				\|x\|_{P_{s(e)}}
				\}.
		\end{equation}
		Therefore, since $\mu_l^e$ is uniform on $\mathcal{E}$, $\mu^x(\Omega) \leq \pi_l(V(\omega_N))|\mathcal{E}|$. Finally, following Proposition~\ref{PAC1}, \eqref{measure_omega} holds with probability at least $1 - \beta$.
		
		\item \label{item1}
		Now, for all $e \in \mathcal{E}$, we perform a change of variable.
		For any $e \in \mathcal{E}$, consider the \emph{Cholesky} factorization of $P_{s(e)}$ and $P_{t(e)}$, that is $P_{s(e)} = L_{s(e)}^TL_{s(e)}$ and $P_{t(e)} = L_{t(e)}^TL_{t(e)}$. Now, for all $e \in \mathcal{E}$, let $B_e = L_{t(e)} A_e L_{s(e)}^{-1}$. Equation~\eqref{eqwithomegaquad} becomes
		\begin{equation}
		\begin{aligned} 
			&\forall e \in \mathcal{E}, \forall x \in L_{s(e)}(\mathbb{S}^n \setminus \Omega) : 
			x^TB_e^TB_ex \leq \overline{\lambda}_e^2 x^Tx \\
			\iff &\forall e \in \mathcal{E}, \forall x \in \Pi_{\mathbb{S}^n}(L_{s(e)}(\mathbb{S}^n \setminus \Omega)) : 
			x^TB_e^TB_ex \leq \overline{\lambda}_e^2 x^Tx \\
			\iff & \forall e \in \mathcal{E}, \forall x \in \mathbb{S}^n \setminus \Pi_{\mathbb{S}^n}(L_{s(e)}(\Omega)) : 
			x^TB_e^TB_ex \leq \overline{\lambda}_e^2 x^Tx,
		\end{aligned}
		\end{equation}
		where the second equivalence holds by homogeneity of the system. For the sake of readibility, we write $\Omega'_e = \Pi_{\mathbb{S}^n}(L_{s(e)}(\Omega))$.
		\item \label{item2}
		We now link $\mu^x(\Omega)$ and $\mu^x(\Omega'_e)$. Using a similar reasoning as in \citep[Theorem~15]{kenanian_data_2019}, for all $e \in \mathcal{E}$, the following holds: 
		\begin{equation} \label{omega_prime_upper}
			\mu^x(\Omega'_e) \leq \sqrt{\frac{\det(P_{s(e)})}{\lambda^{\min}(P_{s(e)})^n}}\mu^x(\Omega) = \kappa(P_{s(e)})\mu^x(\Omega).
		\end{equation}
		This yields 
		\begin{equation} \label{eqwithomegaprime}
			\forall e \in \mathcal{E}, \forall x \in \mathbb{S}^n \setminus \Omega'_e : 
			x^TB_e^TB_ex \leq \overline{\lambda}_e^2 x^Tx, 
		\end{equation}
		where, for all $e \in \mathcal{E}$, $\mu^x(\Omega'_e) \leq \kappa(P_{s(e)}) \mu^x(\Omega)$.
		\item \label{item3}
		We now look for the largest sphere included in $\text{conv}(\mathbb{S}^n \setminus \Omega'_{e})$. The radius of this sphere is denoted by $\alpha$. First, \eqref{eqwithomegaprime} is equivalent to
		\begin{equation} \label{illustration}
			\forall e \in \mathcal{E} : B_e(\mathbb{S}^n \setminus \Omega'_e) \subset \overline{\lambda}_e \mathbb{B}^n. 
		\end{equation}
		Now, \citep[Property~2]{kenanian_data_2019} implies that 
		\begin{equation}
			\forall e \in \mathcal{E} : B_e(\text{conv}(\mathbb{S}^n \setminus \Omega'_e)) \subset \text{conv}(B_e(\mathbb{S}^n \setminus \Omega'_e)) \subset \overline{\lambda}_e \mathbb{B}^n.
		\end{equation}
		Now, from the definition of $\alpha$, for all $e \in \mathcal{E}$, we have that $\alpha\mathbb{S}^n \subset \text{conv}(\mathbb{S}^n \setminus \Omega'_e)$, which implies that $B_e(\alpha \mathbb{S}^n) \subset B_e(\text{conv}(\mathbb{S}^n \setminus \Omega'_e)) \subset \overline{\lambda}_e\mathbb{B}^n$. This yields 
		\begin{equation} 
		\forall e \in \mathcal{E} : B_e(\mathbb{S}^n) \subset (\overline{\lambda}_e / \alpha) \mathbb{B}^n.
		\end{equation}
		Finally, following \citep[Proposition~13]{kenanian_data_2019} and step \ref{item2},
		\begin{equation}
			\alpha \geq \delta(\mu^x(\Omega'_e) / 2 ) \geq \delta( \kappa(P_{s(e)}) \mu^x(\Omega) / 2).
		\end{equation}
		\end{enumerate}
		An illustration of Equation~\eqref{illustration}, as well as the quantity $\alpha$ can be found in Figure~\ref{illustration_proof}.
		\begin{figure}[h!]
			\centering
			\input{illustrations/proof_illustration.tex}
			\caption{Illustration of the proof of Theorem~\ref{maintheorem1}. Each node of $\mathcal{G}^{(l)}$ is the starting node of some edge. Suppose a node $s \in \mathcal{U}$ is the starting node of two edges $e_1$ and $e_2$. For each edge, there is only one corresponding set $\Omega' = \Omega'_{e_1} = \Omega'_{e_2}$ such that, for $e \in \{e_1, e_2\}$, $B_e(\mathbb{S}^n \setminus \Omega') \subset \overline{\lambda}_{e}\mathbb{B}^n$. The quantity $\alpha$ is the radius of the largest sphere inscribed in $\text{conv}(\mathbb{S}^n \setminus \Omega')$.}
			\label{illustration_proof}
		\end{figure}
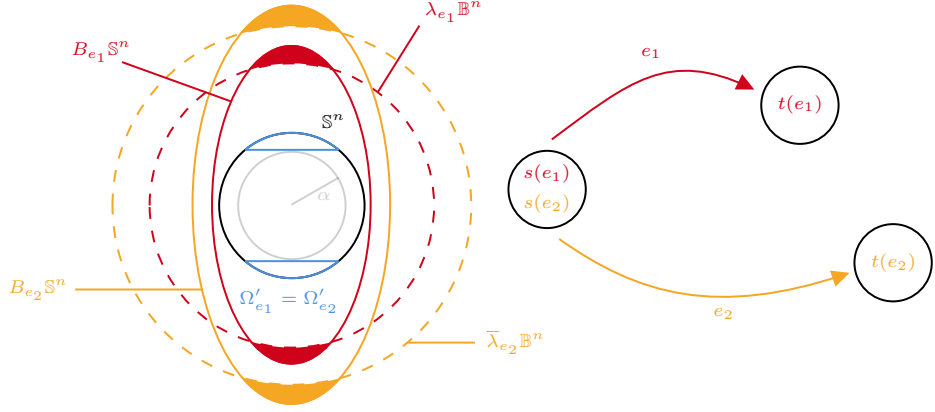
		
		We now summarize. We first define 
		\begin{equation}
			\overline{\rho}_e
			= \frac{\overline{\lambda}_e}{\delta(\varepsilon(\beta, N) |\mathcal{E}(G^{(l)})| \kappa(P^*_{s(e)}(\omega_N)) / 2)}.
		\end{equation}
		Following steps \ref{item0}, \ref{item1}, \ref{item2}, \ref{item3}, the following holds with probability at least $1 - \beta$: 
		\begin{equation} \label{inclusion_circle}
			\forall e \in \mathcal{E} : B_e(\mathbb{S}^n) \subset \overline{\rho}_e \mathbb{B}^n.
		\end{equation}
		From the definition of $B_e$, \eqref{inclusion_circle} is equivalent to 
		\begin{equation}
			\forall e \in \mathcal{E} : A_e \left( L_{s(e)}^{-1}(\mathbb{S}^n) \right) \subset \overline{\rho}_e \left( L_{t(e)}^{-1}(\mathbb{B}^n) \right), 
		\end{equation}
		which implies that
		\begin{equation} \label{levelsetsinclusion}
			\forall e \in \mathcal{E} : A_e \left( L_{s(e)}^{-1}(\mathbb{B}^n) \right) \subset \overline{\rho}_e \left( L_{t(e)}^{-1}(\mathbb{B}^n) \right).
		\end{equation}
		
		Since $\overline{\rho}_e \leq \overline{\rho}(\omega_N)^l$ for all $e \in \mathcal{E}$ following step~\ref{item0}, the next condition holds with probability at least $1 - \beta$: 
		\begin{equation} \label{levelsetsinclusionGlobal}
			\forall e \in \mathcal{E} : A_e \left( L_{s(e)}^{-1}(\mathbb{B}^n) \right) \subset \overline{\rho}(\omega_N)^l \left( L_{t(e)}^{-1}(\mathbb{B}^n) \right).
		\end{equation}
		Furthermore, since $L^{-1}_{s(e)}(\mathbb{B}^n)$ and $L^{-1}_{t(e)}(\mathbb{B}^n)$ are respectively the one-level set of the quadratic norms $\| \cdot \|_{P_{s(e)}}$ and $\| \cdot \|_{P_{s(e)}}$, condition \eqref{levelsetsinclusionGlobal} is equivalent to 
		\begin{equation}
			\forall e \in \mathcal{E}, \forall x \in \mathbb{S}^n : \|A_ex\|_{P_{t(e)}} \leq \overline{\rho}(\omega_N)^l \|x\|_{P_{s(e)}}, 
		\end{equation}
		which yields that $\gamma^*(\mathcal{G}^{(l)}, \mathcal{A}_{\mathcal{G}}^{(l)}) \leq \overline{\rho}(\omega_N)^l$ following definition~\eqref{quadraticupperboundCJSR}. Therefore, it holds that 
		\begin{equation}
			\rho(\mathcal{G}, \mathcal{A}) \leq \gamma^*(\mathcal{G}^{(l)}, \mathcal{A}_{\mathcal{G}}^{(l)})^{1/l} \leq \overline{\rho}(\omega_N)
		\end{equation}
		with probability at least $1 - \beta$, and the proof is completed.
	\end{proof}
	
\begin{rem} \label{rem_lower_bound_instead}
	In parallel to \citep[Remark~16]{kenanian_data_2019}, in the proof of Theorem~\ref{maintheorem1}, in \eqref{omega_prime_upper} for all edge $e \in \mathcal{E}$, one can compute a lower bound on $\mu^x(\mathbb{S}^n \setminus \Omega'_e)$, the measure of the complement of $\Omega'_e$, instead of finding an upper bound on $\mu^x(\Omega'_e)$. In this case, 
\begin{equation}
	\mu^x(\mathbb{S}^n \setminus \Omega'_e)
	\geq \sqrt{\frac{\det(P_{s(e)})}{\lambda^{\max}(P_{s(e)})^n}}\mu^x(\mathbb{S}^n \setminus \Omega)
	= \sqrt{\frac{\det(P_{s(e)})}{\lambda^{\max}(P_{s(e)})^n}}(1 - \mu^x(\Omega)).
\end{equation}
As a consequence, one can replace $\varepsilon(\beta, N)|\mathcal{E}(G^{(l)})|\sqrt{\det(P^*_{s(e)}(\omega_N))/\lambda^{\min}(P^*_{s(e)}(\omega_N))^n}$ in the expression of the bound \eqref{maintheorem1-bound} by
\begin{equation}
	1 - \left( (1 - \varepsilon(\beta, N)|\mathcal{E}(G^{(l)})|)\sqrt{\frac{\det(P^*_{s(e)}(\omega_N))}{\lambda^{\max}(P^*_{s(e)}(\omega_N))^n}}\right), 
\end{equation}
therefore providing an alternative bound. The user can then take the tightest bound as the best probabilistic upper bound on the CJSR.
\end{rem}
	
Theorem~\ref{maintheorem1} provides a probabilistic stability guarantee. For a finite sample set $\omega_N$, and for some confidence level $\beta \in (0, 1)$, if one computes $\overline{\rho}(\omega_N) < 1$, then, following Proposition~\ref{CJSRstability}, one has a certificate of stability for the observed CSLS with probability of at least $1 - \beta$. Moreover, it can be proven that, if one assumes no noise on the observations, for $l$ and $N$ going to infinity, the bound \eqref{maintheorem1-bound} is arbitrarily close to the true CJSR. As this result is a consequence of the well known model-based result \eqref{l-lifting-eq}, and follows a similar reasoning as in the proof of \citep[Theorem~17]{kenanian_data_2019}, we do not provide it in this paper.

%

%% file: illustrations/proof_illustration.tex
\tikzset{every picture/.style={line width=0.75pt}} 

\begin{tikzpicture}[x=0.75pt,y=0.75pt,yscale=-1,xscale=1]

\draw  [color={rgb, 255:red, 245; green, 166; blue, 35 }  ,draw opacity=1 ][fill={rgb, 255:red, 245; green, 166; blue, 35 }  ,fill opacity=1 ] (156.17,124.63) .. controls (156.17,69.13) and (178.48,24.13) .. (206,24.13) .. controls (233.52,24.13) and (255.83,69.13) .. (255.83,124.63) .. controls (255.83,180.14) and (233.52,225.13) .. (206,225.13) .. controls (178.48,225.13) and (156.17,180.14) .. (156.17,124.63) -- cycle ;
\draw  [color={rgb, 255:red, 245; green, 166; blue, 35 }  ,draw opacity=1 ][fill={rgb, 255:red, 255; green, 255; blue, 255 }  ,fill opacity=1 ][dash pattern={on 4.5pt off 4.5pt}] (115.8,125) .. controls (115.8,75.04) and (156.3,34.53) .. (206.27,34.53) .. controls (256.23,34.53) and (296.73,75.04) .. (296.73,125) .. controls (296.73,174.96) and (256.23,215.47) .. (206.27,215.47) .. controls (156.3,215.47) and (115.8,174.96) .. (115.8,125) -- cycle ;
\draw  [color={rgb, 255:red, 208; green, 2; blue, 27 }  ,draw opacity=1 ][fill={rgb, 255:red, 208; green, 2; blue, 27 }  ,fill opacity=1 ] (166,124.63) .. controls (166,80.36) and (183.91,44.47) .. (206,44.47) .. controls (228.09,44.47) and (246,80.36) .. (246,124.63) .. controls (246,168.91) and (228.09,204.8) .. (206,204.8) .. controls (183.91,204.8) and (166,168.91) .. (166,124.63) -- cycle ;
\draw  [color={rgb, 255:red, 208; green, 2; blue, 27 }  ,draw opacity=1 ][fill={rgb, 255:red, 255; green, 255; blue, 255 }  ,fill opacity=1 ][dash pattern={on 4.5pt off 4.5pt}] (134.53,125) .. controls (134.53,85.38) and (166.65,53.27) .. (206.27,53.27) .. controls (245.88,53.27) and (278,85.38) .. (278,125) .. controls (278,164.62) and (245.88,196.73) .. (206.27,196.73) .. controls (166.65,196.73) and (134.53,164.62) .. (134.53,125) -- cycle ;
\draw   (169.6,125) .. controls (169.6,104.75) and (186.02,88.33) .. (206.27,88.33) .. controls (226.52,88.33) and (242.93,104.75) .. (242.93,125) .. controls (242.93,145.25) and (226.52,161.67) .. (206.27,161.67) .. controls (186.02,161.67) and (169.6,145.25) .. (169.6,125) -- cycle ;
\draw  [color={rgb, 255:red, 208; green, 2; blue, 27 }  ,draw opacity=1 ] (166,124.63) .. controls (166,80.36) and (183.91,44.47) .. (206,44.47) .. controls (228.09,44.47) and (246,80.36) .. (246,124.63) .. controls (246,168.91) and (228.09,204.8) .. (206,204.8) .. controls (183.91,204.8) and (166,168.91) .. (166,124.63) -- cycle ;
\draw  [draw opacity=0] (182.7,96.91) .. controls (189.07,91.56) and (197.29,88.33) .. (206.27,88.33) .. controls (215.24,88.33) and (223.46,91.56) .. (229.84,96.91) -- (206.27,125) -- cycle ; \draw  [color={rgb, 255:red, 74; green, 144; blue, 226 }  ,draw opacity=1 ] (182.7,96.91) .. controls (189.07,91.56) and (197.29,88.33) .. (206.27,88.33) .. controls (215.24,88.33) and (223.46,91.56) .. (229.84,96.91) ;  
\draw  [draw opacity=0] (229.84,153.09) .. controls (223.46,158.44) and (215.24,161.67) .. (206.27,161.67) .. controls (197.29,161.67) and (189.07,158.44) .. (182.7,153.09) -- (206.27,125) -- cycle ; \draw  [color={rgb, 255:red, 74; green, 144; blue, 226 }  ,draw opacity=1 ] (229.84,153.09) .. controls (223.46,158.44) and (215.24,161.67) .. (206.27,161.67) .. controls (197.29,161.67) and (189.07,158.44) .. (182.7,153.09) ;  
\draw [color={rgb, 255:red, 208; green, 2; blue, 27 }  ,draw opacity=1 ]   (337.37,91.66) .. controls (376.57,62.26) and (396.05,46.15) .. (436.97,65.28) ;
\draw [shift={(439.5,66.5)}, rotate = 206.21] [fill={rgb, 255:red, 208; green, 2; blue, 27 }  ,fill opacity=1 ][line width=0.08]  [draw opacity=0] (8.93,-4.29) -- (0,0) -- (8.93,4.29) -- cycle    ;
\draw [color={rgb, 255:red, 245; green, 166; blue, 35 }  ,draw opacity=1 ]   (341.11,141.92) .. controls (381.04,170.3) and (420.7,181.44) .. (485.53,159.67) ;
\draw [shift={(487.5,159)}, rotate = 161.05] [fill={rgb, 255:red, 245; green, 166; blue, 35 }  ,fill opacity=1 ][line width=0.08]  [draw opacity=0] (8.93,-4.29) -- (0,0) -- (8.93,4.29) -- cycle    ;
\draw  [color={rgb, 255:red, 245; green, 166; blue, 35 }  ,draw opacity=1 ] (156.17,124.63) .. controls (156.17,69.13) and (178.48,24.13) .. (206,24.13) .. controls (233.52,24.13) and (255.83,69.13) .. (255.83,124.63) .. controls (255.83,180.14) and (233.52,225.13) .. (206,225.13) .. controls (178.48,225.13) and (156.17,180.14) .. (156.17,124.63) -- cycle ;
\draw [color={rgb, 255:red, 74; green, 144; blue, 226 }  ,draw opacity=1 ]   (182.7,96.91) -- (229.84,96.91) ;
\draw [color={rgb, 255:red, 74; green, 144; blue, 226 }  ,draw opacity=1 ]   (182.7,153.09) -- (229.84,153.09) ;
\draw  [color={rgb, 255:red, 0; green, 0; blue, 0 }  ,draw opacity=0.19 ] (179.18,125) .. controls (179.18,110.04) and (191.31,97.92) .. (206.27,97.92) .. controls (221.22,97.92) and (233.35,110.04) .. (233.35,125) .. controls (233.35,139.96) and (221.22,152.08) .. (206.27,152.08) .. controls (191.31,152.08) and (179.18,139.96) .. (179.18,125) -- cycle ;
\draw   (315.5,116.9) .. controls (315.5,106.13) and (324.23,97.4) .. (335,97.4) .. controls (345.77,97.4) and (354.5,106.13) .. (354.5,116.9) .. controls (354.5,127.67) and (345.77,136.4) .. (335,136.4) .. controls (324.23,136.4) and (315.5,127.67) .. (315.5,116.9) -- cycle ;
\draw [color={rgb, 255:red, 0; green, 0; blue, 0 }  ,draw opacity=0.19 ]   (206,124.63) -- (229.8,111) ;
\draw [color={rgb, 255:red, 208; green, 2; blue, 27 }  ,draw opacity=1 ]   (249.6,67.8) -- (270.6,34.2) ;
\draw [color={rgb, 255:red, 245; green, 166; blue, 35 }  ,draw opacity=1 ]   (264.8,194.2) -- (298.6,194.2) ;
\draw [color={rgb, 255:red, 208; green, 2; blue, 27 }  ,draw opacity=1 ]   (175.77,72.53) -- (120.6,52.6) ;
\draw [color={rgb, 255:red, 245; green, 166; blue, 35 }  ,draw opacity=1 ]   (97,167.4) -- (161.4,167.4) ;
\draw   (443,74.4) .. controls (443,63.63) and (451.73,54.9) .. (462.5,54.9) .. controls (473.27,54.9) and (482,63.63) .. (482,74.4) .. controls (482,85.17) and (473.27,93.9) .. (462.5,93.9) .. controls (451.73,93.9) and (443,85.17) .. (443,74.4) -- cycle ;
\draw   (490,153.9) .. controls (490,143.13) and (498.73,134.4) .. (509.5,134.4) .. controls (520.27,134.4) and (529,143.13) .. (529,153.9) .. controls (529,164.67) and (520.27,173.4) .. (509.5,173.4) .. controls (498.73,173.4) and (490,164.67) .. (490,153.9) -- cycle ;

\draw (178.4,165.6) node [anchor=north west][inner sep=0.75pt]  [font=\scriptsize,color={rgb, 255:red, 74; green, 144; blue, 226 }  ,opacity=1 ]  {$\Omega '_{e_{1}} =\Omega '_{e_{2}}$};
\draw (321.6,102.8) node [anchor=north west][inner sep=0.75pt]  [font=\scriptsize,color={rgb, 255:red, 255; green, 255; blue, 255 }  ,opacity=1 ]  {$\textcolor[rgb]{0.82,0.01,0.11}{s( e_{1})}$};
\draw (321.6,117.8) node [anchor=north west][inner sep=0.75pt]  [font=\scriptsize,color={rgb, 255:red, 255; green, 255; blue, 255 }  ,opacity=1 ]  {$\textcolor[rgb]{0.96,0.65,0.14}{s( e_{2})}$};
\draw (449.9,67.7) node [anchor=north west][inner sep=0.75pt]  [font=\scriptsize,color={rgb, 255:red, 255; green, 255; blue, 255 }  ,opacity=1 ]  {$\textcolor[rgb]{0.82,0.01,0.11}{t( e_{1})}$};
\draw (497,147.5) node [anchor=north west][inner sep=0.75pt]  [font=\scriptsize,color={rgb, 255:red, 255; green, 255; blue, 255 }  ,opacity=1 ]  {$\textcolor[rgb]{0.96,0.65,0.14}{t( e_{2})}$};
\draw (217.5,116.7) node [anchor=north west][inner sep=0.75pt]  [font=\scriptsize, opacity=0.19 ]  {$\alpha $};
\draw (220.1,77.6) node [anchor=north west][inner sep=0.75pt]  [font=\scriptsize  ,opacity=1 ]  {$\mathbb{S}^{n}$};
\draw (271.8,18.7) node [anchor=north west][inner sep=0.75pt]  [font=\scriptsize]  {$\textcolor[rgb]{0.82,0.01,0.11}{\overline{\lambda}_{e_1}\mathbb{B}^{n}}$};
\draw (303,186.4) node [anchor=north west][inner sep=0.75pt]  [font=\scriptsize,color={rgb, 255:red, 245; green, 166; blue, 35 }  ,opacity=1 ]  {$\textcolor[rgb]{0.96,0.65,0.14}{\overline{\lambda }_{e_2}\mathbb{B}^{n}}$};
\draw (93.6,40.7) node [anchor=north west][inner sep=0.75pt]  [font=\scriptsize]  {$\textcolor[rgb]{0.82,0.01,0.11}{B_{e_{1}}\mathbb{S}^{n}}$};
\draw (61.8,161) node [anchor=north west][inner sep=0.75pt]  [font=\scriptsize]  {$\textcolor[rgb]{0.96,0.65,0.14}{B_{e_{2}}\mathbb{S}^{n}}$};
\draw (381,43.9) node [anchor=north west][inner sep=0.75pt]  [font=\scriptsize]  {$\textcolor[rgb]{0.82,0.01,0.11}{e_{1}}$};
\draw (417.5,175.4) node [anchor=north west][inner sep=0.75pt]  [font=\scriptsize]  {$\textcolor[rgb]{0.96,0.65,0.14}{e_{2}}$};

\end{tikzpicture}

%% file: content/4_continuous_states.tex
\section{Learning stability guarantees by observing continuous states only}
\label{continuousStates}

In some situations, it can be difficult to observe the discrete part of the dynamics. In this section, we show that one may still be able to derive a probabilistic upper bounds for the underlying CSLS under reasonable assumptions. For this purpose, we first present a model-based result on the CJSR that allows to study the stability of a CSLS via the analysis of an auxiliary arbitrary switching linear system. Then, we formulate the problem of finding a probabilistic upper bound on the CJSR from continuous states. Finally, we provide a solution based on chance-constrained optimization.

\subsection{Reducing constrained systems to arbitrary systems}
In this subsection we provide a general result, of application in the wider context of model-based analysis of CSLS, which will be useful for our purpose.

Following Remark~\ref{flower}, the quantity $\rho( \mathcal{F}_m, \mathcal{A}_{\mathcal{G}}^{(l)})$ is the (unconstrained) joint spectral radius of the ASLS defined on the (constrained) set of matrices $\mathcal{A}_{\mathcal{G}}^{(l)}$. Proposition~\ref{lifting} states that an upper bound on the initial CJSR can be derived from this quantity. In the next subsection, we will use the continuous data to which we have access to approximate it, and thereby obtain probabilistic stability guarantees.
\begin{prop}
\label{lifting}
    Consider a CSLS $S(\mathcal{G}, \mathcal{A})$. For any length $l \in \mathbb{N}_{>0}$, the following holds: 
    \begin{equation}
    \label{finiteIneq}
        \rho(\mathcal{G}, \mathcal{A}) \leq \rho(\mathcal{F}_m, \mathcal{A}_{\mathcal{G}}^{(l)})^{1/l}.
    \end{equation}
    Moreover, asymptotic equality holds, that is
    \begin{equation}
    \label{asymEq}
        \rho(\mathcal{G}, \mathcal{A}) = \lim_{l \to \infty} \rho( \mathcal{F}_m, \mathcal{A}_{\mathcal{G}}^{(l)})^{1/l}.
    \end{equation}
\end{prop}
\begin{proof}
	First we prove inequality~\eqref{finiteIneq}. Since for all $t > 0$ and $l > 0$, $\mathcal{A}_{\mathcal{G}}^{(lt)} \subseteq \left(\mathcal{A}_{\mathcal{G}}^{(l)}\right)^t$, Definition~\ref{CJSRdef} and Remark~\ref{flower} yields 
\begin{equation}
\begin{aligned}
	\rho(\mathcal{F}_m, \mathcal{A}_{\mathcal{G}}^{(l)}) 	&= \lim_{t \to \infty} \max \left\{ \| A \|^{1/t}: \, A \in \left(\mathcal{A}_{\mathcal{G}}^{(l)}\right)^t \right\} \\
					&\geq \lim_{t \to \infty} \max \left\{ \| A \|^{1/t}: \, A \in \mathcal{A}_{\mathcal{G}}^{(lt)} \right\}.			
\end{aligned}					
\end{equation}
Letting $k = lt$ and, again, using Definition~\ref{CJSRdef}, the inequality becomes
\begin{equation} \label{proof_geq}
\begin{aligned}
	\rho(\mathcal{F}_m, \mathcal{A}_{\mathcal{G}}^{(l)}) 	&\geq \lim_{t \to \infty} \max \left\{ \| A \|^{l/k}: \, A \in \mathcal{A}_{\mathcal{G}}^{(k)} \right\} \\
					&= \left( \lim_{k \to \infty} \max \left\{ \| A \|^{1/k}: \, A \in \mathcal{A}_{\mathcal{G}}^{(k)} \right\} \right)^l \\
					&= \rho(\mathcal{G}, \mathcal{A})^l, 
\end{aligned}
\end{equation}
which is the desired result. Now we prove asymptotic equality~\eqref{asymEq}. First, given the description of the CJSR in \eqref{CJSRotherdef}, for any $l > 0$,
\begin{equation}
\begin{aligned}
	\rho(\mathcal{F}_m, \mathcal{A}_{\mathcal{G}}^{(l)}) &\leq \max \left\{ \| A \|: \, A \in \mathcal{A}_{\mathcal{G}}^{(l)} \right\} \\
	\iff \rho(\mathcal{F}_m, \mathcal{A}_{\mathcal{G}}^{(l)})^{1/l} &\leq \max \left\{ \| A \|^{1/l}: \, A \in \mathcal{A}_{\mathcal{G}}^{(l)} \right\}.
\end{aligned}
\end{equation}
Taking the limit in both sides, it gives $\lim_{l \to \infty} \rho(\mathcal{F}_m, \mathcal{A}_{\mathcal{G}}^{(l)})^{1/l} \leq \rho(\mathcal{G}, \mathcal{A})$. Combining with \eqref{proof_geq}, it gives the desired result.
\end{proof}



\subsection{Problem formulation}

In this section, we formulate the problem of finding probabilistic stability guarantees only from continuous information. For some unknown labeled graph $\mathcal{G} = (G, \sigma)$ and set of matrices $\mathcal{A} = \{A_1, \dots, A_m\}$, consider the CSLS $S(\mathcal{G}, \mathcal{A})$, for which one can sample data containing only information about the continuous states. 

More formally, consider the $l$-product lift of the observed system, that is $S(\mathcal{G}^{(l)}, \mathcal{A}_{\mathcal{G}}^{(l)})$, where $\mathcal{G} = (G^{(l)}, \sigma^{(l)})$ according to Definition~\ref{product_lifting}. One \emph{sample of length $l$} is defined as the triplet $(x, \sigma, w)$, where 
\begin{itemize}
	\item $x \in \mathbb{S}^n$ is an initial continuous state: $x(0) = x$ in \eqref{hybrid_dynamics}, sampled according to the uniform measure $\mu^{x}$ as described in Section~\ref{hybrid_pf}.
	\item $\sigma$ is such that $\sigma = \sigma^{(l)}(e)$, for some unknown edge $e \in \mathcal{E}(G^{(l)})$. $\sigma$ is sampled according to the probability measure $\mu^{\sigma}_l$, which is assumed to be uniform on $\mathcal{L}_{G, l}$.
	\item $w \in B^n[W]$ is a noise on the continuous state-space, where $W > 0$ is a known upper bound, sampled according to the probability measure $\mu^w$ as described in Section~\ref{hybrid_pf}.
\end{itemize}
In the continuous framework, with $\Delta_{\text{cont}} = \mathbb{S}^n \times [m] \times B^n[W]$, the sample set is therefore defined as 
\begin{equation} \label{cont_observations}
	\omega_{N, \text{cont}} = \left\{\left(x_i, \sigma_i, w_i\right)\right\}_{i = 1, \dots, N} \subset \Delta_{\text{cont}}, 
\end{equation}
where each element $(x_i, \sigma_i, w_i)$ is sampled according to the product probability measure $\pi_{l, \text{cont}} = \mu^x \times \mu^\sigma_l \times \mu^w$.

\begin{rem}
In parallel to Remark~\ref{hybrid_link_samples_observations}, but in the framework of this section, we define an \emph{observation} as a continuous state $x \in \mathbb{R}$. The link between the sample $(x_i, \sigma_i, w_i)$ and the pair of observations $x_i, y_i$, is that $y_i = A_{\sigma_i}x_i + w_i$. 
\end{rem}

Our second problem is stated as follows.
\begin{prob} \label{problem2}
	Consider an unknown CSLS $S(\mathcal{G}, \mathcal{A})$. For a fixed length $l \in \mathbb{Z}_{>0}$ and number $N \in \mathbb{Z}_{>0}$ of samples, consider the sample set $\omega_{N, \text{cont}}$ as described in \eqref{cont_observations}. For some user-defined confidence level of at least $\beta \in (0, 1)$, and knowing $\mathcal{U}(G)$ and $|\mathcal{L}_{\mathcal{G}, l}|$, the number of words of length $l$ generated by $\mathcal{G}$, provide an upper bound $ \rho(\mathcal{G}, \mathcal{A}) \leq  \overline{\rho}$.
\end{prob}

\subsection{Results}

To answer to Problem~\ref{problem2}, we will now approximate a quadratic Lyapunov function for the unknown arbitrary system $S(\mathcal{F}_m, \mathcal{A}_{\mathcal{G}}^{(l)})$ following Proposition~\ref{lifting}. In order to do so, let us introduce the following \emph{sampled program}:
\begin{subequations}
\label{scenario_program_2}
\begin{align}
	\mathcal{P}(\omega_{N, \text{cont}}): &\min_{\substack{\lambda \in \mathbb{R}_{\geq 0} \\ P \in \mathbb{R}^{n \times n}}} \left(\lambda, c(P)\right) \\
	\text{s.t. } 
	&P \in \mathcal{X} = \{ P \in \mathbb{R}^{n \times n} \, | \, P \succeq I, \| P \|_F \leq C \},  \\
	&\forall (x, \sigma, w) \in \omega_{N, \text{cont}}: 
	\left\| 
	A_{\sigma} x + w
	\right\|_{P}
	\leq 
	\lambda^{l} 
	\left\|
	x
	\right\|_{P}, 
\end{align}
\end{subequations}
where $c: \mathcal{X} \to \mathbb{R} : P \mapsto \|P\|_F$, and $C$ is sufficiently large. We denote the optimal solutions of the sampled program~\eqref{scenario_program_1} by $\lambda^*(\omega_{N, \text{cont}})$ and $P^*(\omega_{N, \text{cont}})$. 

We first present the following chance-constrained result. 

\begin{prop}
\label{PAC3_cont}
	Consider an unknown CSLS $S(\mathcal{G}, \mathcal{A})$. For any given number of samples $N \in \mathbb{Z}_{>0}$ such that $N \geq d = n(n+1)/2$, and for $l \in \mathbb{N}_{> 0}$, let $\omega_{N, \text{cont}} \subset \Delta_{\text{cont}}$ be the sample set as described in \eqref{cont_observations}. Given $\lambda^*(\omega_{N, \text{cont}})$ and $P^*(\omega_{N, \text{cont}})$ the solutions of the sampled program~\eqref{scenario_program_2}, let 	
	\begin{equation} \label{viol_cont}
		V(\omega_{N, \text{cont}})
			= \left\{ (x, \sigma, w) \in \Delta_{\text{cont}} \, | \,
			\left\|A_\sigma x
			\right\|_{P^*(\omega_{N, \text{cont}})}
			>
			\overline{\lambda}(\omega_{N, \text{cont}})
			\left\|x\right\|_{P^*(\omega_{N, \text{cont}})}
		\right\}, 
	\end{equation}
	where 
	\begin{equation} \label{lambda_with_noise_cont}
		\overline{\lambda}(\omega_{N, \text{cont}}) = \lambda^*(\omega_{N, \text{cont}})^l + \sqrt{\frac{\lambda^{\max}\left(P^*(\omega_{N, \text{cont}}\right))}{\lambda^{\min}\left(P^*(\omega_{N, \text{cont}}))\right)}} W.
	\end{equation}
	Then, for any confidence level $\beta \in (0, 1)$, the following holds: 
	\begin{equation} \label{PAC_bound_cont}
		\pi_{l, \text{cont}}^N\left(
			\left\{
				\omega_{N, \text{cont}} \subset \Delta_{\text{cont}} \, 
				| \, \pi_{l, \text{cont}}\left(
					V(\omega_{N, \text{cont}})
				\right) \leq \varepsilon(\beta, N)
			\right\}
		\right) \geq 1 - \beta, 
	\end{equation} 
	where $\varepsilon(\beta, N)$ is defined in \eqref{vareps}.
\end{prop}
\begin{proof}
	For the sake of readability, we note $\lambda = \lambda^*(\omega_{N, \text{cont}})$, $P = P^*(\omega_{N, \text{cont}})$ and $\overline{\lambda} = \overline{\lambda}(\omega_{N, \text{cont}})$. In the same fashion as in the proof of Proposition~\ref{PAC1}, we first prove that 
	\begin{equation} \label{PAC_result_cont}
		\pi^N_{l, \text{cont}}(\{
			\omega_{N, \text{cont}} \subset \Delta_{\text{cont}} \, | \,
			\pi_{l, \text{cont}} \left(\tilde{V}\left(\omega_{N, \text{cont}}\right)\right)
			\leq  
			\varepsilon(\beta, N)
		\}) 
		\geq 
		1 - \beta, 
	\end{equation}
	with 
	\begin{equation}
		\tilde{V}(\omega_{N, \text{cont}}) = 
		\{ 
			(x, \sigma, w) \in \Delta_{\text{cont}} \,|\, 
			\|A_\sigma x + w\|_P 
			>
			\lambda \|x\|_P
		\}. 
	\end{equation}
	Again, \eqref{PAC_result_cont} holds by \citep[Theorem~6]{berger_chance-constrained_2021}, since \eqref{scenario_program_2} is the sampled form of a quasi-linear program, as defined in \citep{berger_chance-constrained_2021}, and that $\omega_{N, \text{cont}}$ is almost surely non-degenerate following Proposition~\ref{non_degenerate_for_sure}. It remains to prove that $V(\omega_{N, \text{cont}}) \subseteq \tilde{V}(\omega_{N, \text{cont}})$, that is, for any $(x, \sigma, w) \in \Delta_{\text{cont}}$, 
	\begin{equation}
		\|A_\sigma x\|_{P} > \overline{\lambda} \|A\|_P
		\implies
		\|A_\sigma x + w \| > \lambda^l \|A\|_P.
	\end{equation}
	The latter holds by applying the exact same argument as in \ref{prop_proof_item2} in the proof of Proposition~\ref{PAC1}.
\end{proof}
%

Now we present the main theorem of this section by leveraging the model-based result obtained in Proposition~\ref{lifting}.

\begin{thm} \label{maintheorem2}
	Consider an unknown CSLS $S(\mathcal{G}, \mathcal{A})$. For any given number of samples $N \in \mathbb{Z}_{>0}$ such that $N \geq d=  n(n+1)/2$, and for $l \in \mathbb{N}_{> 0}$, let $\omega_{N, \text{cont}} \subset \Delta_{\text{cont}}$ be the sample set as described in \eqref{cont_observations}. Given $\lambda^*(\omega_{N, \text{cont}})$ and $P^*(\omega_{N, \text{cont}})$ the solutions of the sampled program~\eqref{scenario_program_2}, for any confidence level $\beta \in (0, 1)$, the following holds: 
	\begin{equation}
		\pi_{l, \text{cont}}^N\left(
			\left\{
				\omega_{N, \text{cont}} \subset \Delta_{\text{cont}} \, 
				| \, \rho(\mathcal{G}, \mathcal{A}) \leq \overline{\rho}(\omega_{N, \text{cont}})
			\right\}
		\right) \geq 1 - \beta, 
	\end{equation} 
	where
	\begin{equation} \label{maintheorem2-bound}
		\overline{\rho}(\omega_{N, \text{cont}}) =  \left( 
			\frac{\overline{\lambda}(\omega_{N, \text{cont}})}{\delta\left(\varepsilon(\beta, N) |\mathcal{L}_{\mathcal{G}, l}| \kappa(P^*(\omega_{N, \text{cont}})) / 2 \right)}
		\right)^{1/l}, 
	\end{equation}
	with $\varepsilon(\beta, N)$, $\kappa(P^*(\omega_{N, \text{cont}}))$, $\delta(\cdot)$ and $\overline{\lambda}(\omega_{N, \text{cont}})$ respectively defined in \eqref{vareps}, \eqref{kappa}, \eqref{delta} and \eqref{lambda_with_noise_cont}.
\end{thm}

\begin{proof}
	We follow a similar reasoning as in the proof of Theorem~\ref{maintheorem1} where any labeled graph $\mathcal{G}$ is replaced by the labeled graph $\mathcal{F}_m$ as described in Remark~\ref{flower}. We will go through the main steps of the proof of Theorem~\ref{maintheorem1}. Again, for the sake of readability, we note $\lambda = \lambda^*(\omega_{N, \text{cont}})$, $P = P^*(\omega_{N, \text{cont}})$ and $\overline{\lambda} = \overline{\lambda}(\omega_{N, \text{cont}})$.
	\begin{enumerate}[(i)]
		\item We first show that, with a probability of at least $1-\beta$, there exists a set $\Omega \subset \mathbb{S}^n$ such that 
		\begin{equation} \label{eqwithomegaquad_cont}
			\forall \sigma \in \mathcal{L}_{\mathcal{G}, l}, \forall x \in \mathbb{S}^n \setminus \Omega :
			\|A_\sigma x \|_{P} \leq \overline{\lambda}\|x\|_P, 
		\end{equation}
		and that satisfies 
		\begin{equation} \label{measure_omega_cont}
			\mu^x(\Omega) \leq \varepsilon(\beta, N)|\mathcal{L}_{\mathcal{G}, l}|.
		\end{equation}
		Again, condition \eqref{eqwithomegaquad_cont} implies that
		\begin{equation} 
			\Omega = \{
				x \in \mathbb{S}^n \,|\, \exists e \in \mathcal{E} :
				\| A_e x \|_{P_{t(e)}}
				>
				\overline{\lambda}_e
				\|x\|_{P_{s(e)}}
				\}.
		\end{equation}
		Therefore, since $\mu_l^\sigma$ is uniform on $\mathcal{L}_{\mathcal{G}, l}$, $\mu^x(\Omega) \leq \pi_{l, \text{cont}}(V(\omega_{N, \text{cont}}))|\mathcal{L}_{\mathcal{G}, l}|$. Finally, following Proposition~\ref{PAC3_cont}, \eqref{measure_omega_cont} holds with probability at least $1 - \beta$.
		
		\item We now perform only one change of variable. Let $P = L^TL$ be the Cholesky decomposition of $P$. For all $\sigma \in \mathcal{L}_{\mathcal{G}, l}$, let $B_\sigma = L A_\sigma L^T$. For the same reasons as in the proof of Theorem~\ref{maintheorem1}, Equation \eqref{eqwithomegaquad_cont} becomes 
		\begin{equation}
			\forall \sigma \in \mathcal{L}_{\mathcal{G}, l}, \forall x \in \mathbb{S}^n \setminus \Omega' :
			x^T B^T_\sigma B_\sigma x \leq \overline{\lambda}^2 x^Tx, 
		\end{equation}
		with $\Omega' = \Pi_{\mathbb{S}^n}(L(\Omega))$.
		
		\item We now link $\mu^x(\Omega)$ and $\mu^x(\Omega')$. Again, for the same reasons as in the proof of Theorem~\ref{maintheorem1}, $\mu^x(\Omega') \leq \kappa(P)\mu^x(\Omega)$ holds.
		
		\item We now look at the largest sphere included in $\text{conv}(\mathbb{S}^n \setminus \Omega')$. The radius of this sphere is noted $\alpha$. Again, a similar reasoning as in the proof of Theorem~\ref{maintheorem1} yields $\alpha \geq \delta(\kappa(P)\mu^x(\Omega)/2)$, which yields
		\begin{equation} \label{last_eq_proof_cont}
			\forall \sigma \in \mathcal{L}_{\mathcal{G}, l} : B_\sigma (\mathbb{S}^n) \subset (\overline{\lambda} / \alpha)\mathbb{B}^n \subset \left(\overline{\rho}(\omega_{N, \text{cont}})^l \right) \mathbb{B}^n.
		\end{equation}
		
	\end{enumerate} 
	Summarizing, \eqref{last_eq_proof_cont} is equivalent to
	\begin{equation} \label{inclusion_cont}
		\forall \sigma \in \mathcal{L}_{\mathcal{G}, l}, \forall x \in \mathbb{S}^n : \|A_\sigma x\|_{P^*(\omega_{N, \text{cont}})} \leq \overline{\rho}(\omega_{N, \text{cont}})^l \|x\|_{P^*(\omega_{N, \text{cont}})}, 
		\end{equation}
	which holds with a probability at least $1 - \beta$. The condition \eqref{inclusion_cont} is equivalent to $\gamma^*(\mathcal{F}_m, \mathcal{A}_{\mathcal{G}}^{(l)}) \leq \overline{\rho}(\omega_{N, \text{cont}})^l$ following definition~\eqref{quadraticupperboundCJSR}. Therefore, following Proposition~\ref{lifting}, one has $\rho(\mathcal{G}, \mathcal{A}) \leq \gamma^*(\mathcal{F}_m, \mathcal{A}_{\mathcal{G}}^{(l)})^{1/l} \leq \overline{\rho}(\omega_{N, \text{cont}})$ with probability at least $1 - \beta$, and the proof is completed.
\end{proof}

\begin{rem}
	In parallel to Remark~\ref{rem_lower_bound_instead}, and for the same reason, one can replace 
\begin{equation}
	\varepsilon(\beta, N)|\mathcal{L}_{\mathcal{G}, l}|\sqrt{\frac{\det(P^*(\omega_N))}{\lambda^{\min}(P^*(\omega_N))^n}}
\end{equation} 
in the expression of the bound \eqref{maintheorem2-bound} by
\begin{equation}
	1 - \left( (1 - \varepsilon(\beta, N)|\mathcal{L}_{\mathcal{G}, l}|)\sqrt{\frac{\det(P^*(\omega_N))}{\lambda^{\max}(P^*(\omega_N))^n}}\right), 
\end{equation}
therefore providing an alternative bound. 
\end{rem}

Theorem~\ref{maintheorem2} provides a probabilistic stability guarantee in the case where one has access only to continuous states. Indeed in comparaison to Theorem~\ref{maintheorem1}, one does not need any information about the starting and terminating nodes for every observation. For a finite sample set $\omega_{N, \text{cont}}$, and for some confidence level $\beta \in (0, 1)$, then if one computes $\overline{\rho}(\omega_{N, \text{cont}}) < 1$, then, following Proposition~\ref{CJSRstability}, one has a certificate of stability for the observed CSLS with probability of at least $1 - \beta$. Again, it can be proven that, if one assumes no noise on the observations, for $l$ and $N$ going to infinity, the bound \eqref{maintheorem2-bound} is arbitrarily close to the true CJSR, thanks to the asymptotic relation \eqref{asymEq}. Again we do not explicitly prove this result in this paper.


%% file: content/5_applications.tex
\section{Application to Networked Control Systems} \label{application}

In this section, we illustrate our method on a concrete example, namely Networked Control Systems (or NCS for short)\footnote{For an introduction, see \citep{hespanha_2007_NCS} and references therein.}. Networked Control Systems are dynamical systems whose plant is physically separated from its controller, so that the feedback has to go through a network. The failures experienced by the network can typically be modelled by a switching rule. For example, consider a linear dynamical system, whose controller is a linear feedback. Considering that packets can be lost in the network, the dynamics of such system is given by 
\begin{equation}
	x(k+1) = \begin{cases}
		(A + BK)x(k)  &\text{if the packet is not lost} \\
		Ax(k) &\text{otherwise.}
	\end{cases}
\end{equation}
An illustration of such dynamical system is given in Figure~\ref{NCS}. 
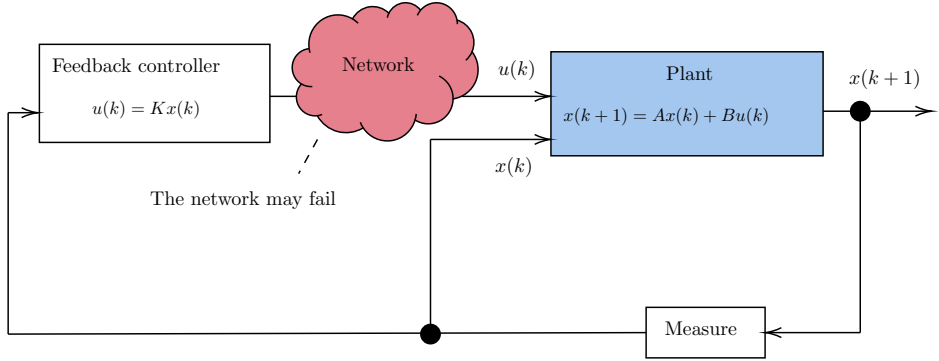
\begin{figure}[h!]
	\centering
	\input{illustrations/NCS.tex}
	\caption{Illustration of a NCS suffering from packet losses.}
	\label{NCS}
\end{figure}
Now, consider that one has more knowledge about the system. For example, consider that packets are never lost more than twice in a row\footnote{Although this constrained is quite simple for the sake of the example, the labeled graph framework allows for way more complex constraints.}. In this case, this system can be modelled as a CSLS $S(\mathcal{G}, \mathcal{A})$, where $\mathcal{A} = \{A_1, A_2\} = \{A + BK, A\}$ and where $\mathcal{G}$ is as depicted in Figure~\ref{twiceinarow}.

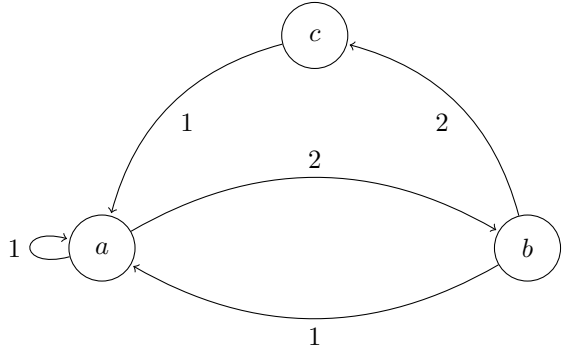
\begin{figure}[h!]
\centering
\begin{tikzpicture}[shorten >=1pt,node distance=4cm,on grid,auto] 
   \node[state] (c) {$c$}; 
   \node[state] (a) [below left = of c] {$a$}; 
   \node[state] (b) [below right = of c] {$b$};
   \path[->]
   	(a) edge[loop left] node {1} (a)
	(a) edge[bend left] node {2} (b)
	(b) edge[bend left] node {1} (a)
	(b) edge[bend right] node {2} (c)
	(c) edge[bend right] node {1} (a)
	;
\end{tikzpicture}
\caption{Labeled graph $\mathcal{G} = (G, \sigma)$ such that its language $\mathcal{L}_{G}$ contains any words except those that have more than two 2 in a row. By using the CSLS $S(\mathcal{G}, \mathcal{A})$ to model Networked Control Systems, it means that the packets cannot be lost more than twice in a row.}
\label{twiceinarow}
\end{figure}

More especially, we consider a two-dimensional NCS with a one-dimensional control signal, whose matrices are given as follows:
\begin{equation}
	A = \begin{pmatrix}
		0.45 &1.08\\ 
		0.36 & 0.09
	\end{pmatrix}, 
	B = \begin{pmatrix}
		0 \\ 1
	\end{pmatrix}, 
	K = \begin{pmatrix}
		-0.42 & -0.36
	\end{pmatrix}.
\end{equation}
Using model-based methods given in \citep{philippe_stability_2016}, we are able to compute $\rho(\mathcal{G}, \mathcal{A}) \approx 0.70697$. 

First, we show the evolution of $\overline{\rho}(\omega_N)$, the bound~\eqref{maintheorem1-bound} in case of noisy observations. In order to do so, for an increasing number of observations of length 1, we compute 20 realizations of $\overline{\rho}(\omega_N)$, such that the noises $w_i$ are independently and uniformly sampled in $B^n[W]$, for different upper bounds $W \in \{0, 0.01, 0.1\}$\footnote{Note that the case $W = 0$ corresponds to a slightly different setting that the one presented in Section~\ref{hybridStates}. However, all results hold in the noise-free setting if the sample set is non-degenerate with probability 1, which one can verifiy in practice.}, and for a confidence level of $\beta = 5\%$. The average is reported in Figure~\ref{exp1}.
\begin{figure}[h!]
	\centering
	\includegraphics[scale=0.5]{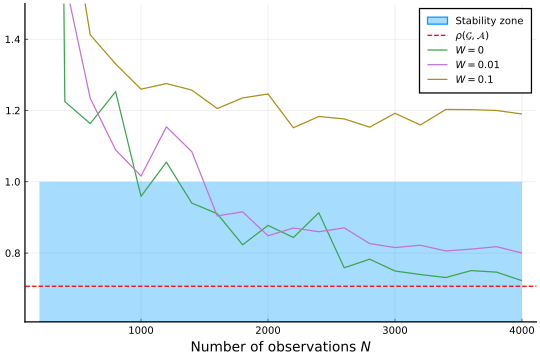}
	\caption{Evolution of the hybrid states bound $\overline{\rho}(\omega_N)$ given in Theorem~\ref{maintheorem1}, for different upper bounds on the norm of the noise $W$. One can observe that a smaller noise allows for a better approximation of the CJSR.}
	\label{exp1}
\end{figure}
We can observe that, in average, the upper bound is still valid, even for noisy observations. In addition, we can see that, for small noises such as $W = 0.01$, our bound is still able to provide stability guarantees with a confidence level of $5\%$. However, with $W = 0.1$, our bound is unable to provide any guarantee after $4000$ observations.

Second, we illustrate that our method is valid for observations of length greater than 1. The evolution of the bound presented in Theorem~\ref{maintheorem1} is presented in Figure~\ref{exp2} for $l \in \{1, 2\}$, and considering no noise. Again, for $\beta = 5\%$, 20 realizations were computed for each considered number of observations, and the average is reported in Figure~\ref{exp2}. 
\begin{figure}[h!]
	\centering
	\includegraphics[scale=0.5]{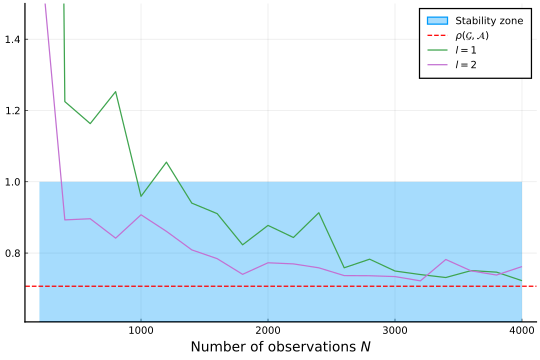}
	\caption{Evolution of the hybrid states bound $\overline{\rho}(\omega_N)$ given in Theorem~\ref{maintheorem1}, for different lengths of observations $l$. Even though the bounds evolute differently, they are both providing stability guarantees.}
	\label{exp2}
\end{figure}
We also obtain stability guarantees for $l = 2$, although we need slightly more samples.

Finally, we show that we can apply our method even if we do not have access to hybrid states, by applying the method described in Section~\ref{continuousStates} for a confidence level of $\beta = 5\%$. The evolution of the average of 20 realizations of $\overline{\rho}(\omega_{N, \text{cont}})$, as presented in Theorem~\ref{maintheorem2}, is given in Figure~\ref{exp3}, for $l \in \{1, 2, 3\}$, where we suppose that the observations are noiseless. The evolution of the bound from hybrid observations is also given as a reference.
\begin{figure}[h!]
	\centering
	\includegraphics[scale=0.5]{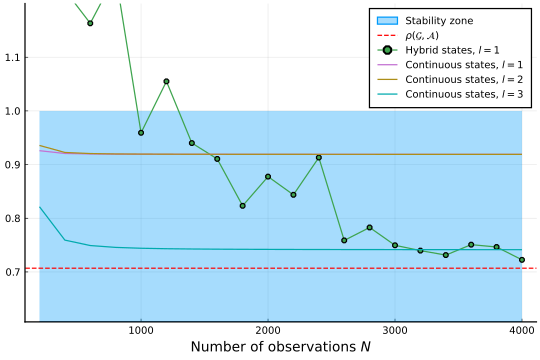}
	\caption{Illustration of the continuous states method. The lines without a marker represent the evolution of $\overline{\rho}(\omega_{N, \text{cont}})$, the bound obtained only from continuous states presented in Theorem~\ref{maintheorem2}, for different lengths $l \in \{1, 2, 3\}$. As a reference, the line with a marker is the evolution of the bound obtained from hybrid states, for $1$-step observations.}
	\label{exp3}
\end{figure}
One can see that the bound computed only from continuous states seems to converge to more conservative values that the hybrid method. This is due to the fact that we are actually approximating $\rho(\mathcal{F}_m, \mathcal{A}_{\mathcal{G}}^{(l)})^{1/l}$, which is greater or equal than $\rho(\mathcal{G}, \mathcal{A})$ following Proposition~\ref{lifting}. However, in this case, probabilistic stability guarantees were still found for each considered value of $l$.

%% file: illustrations/NCS.tex
\scalebox{0.75}{
\tikzset{every picture/.style={line width=0.75pt}} 

\begin{tikzpicture}[x=0.75pt,y=0.75pt,yscale=-1,xscale=1]

\draw    (320,87) -- (378,87) ;
\draw [shift={(380,87)}, rotate = 180] [color={rgb, 255:red, 0; green, 0; blue, 0 }  ][line width=0.75]    (10.93,-3.29) .. controls (6.95,-1.4) and (3.31,-0.3) .. (0,0) .. controls (3.31,0.3) and (6.95,1.4) .. (10.93,3.29)   ;
\draw    (563,97) -- (634,97) ;
\draw [shift={(636,97)}, rotate = 180] [color={rgb, 255:red, 0; green, 0; blue, 0 }  ][line width=0.75]    (10.93,-3.29) .. controls (6.95,-1.4) and (3.31,-0.3) .. (0,0) .. controls (3.31,0.3) and (6.95,1.4) .. (10.93,3.29)   ;
\draw    (299,116) -- (378,116) ;
\draw [shift={(380,116)}, rotate = 180] [color={rgb, 255:red, 0; green, 0; blue, 0 }  ][line width=0.75]    (10.93,-3.29) .. controls (6.95,-1.4) and (3.31,-0.3) .. (0,0) .. controls (3.31,0.3) and (6.95,1.4) .. (10.93,3.29)   ;
\draw    (588,97) -- (588,246) ;
\draw    (588,246) -- (526,246) ;
\draw [shift={(524,246)}, rotate = 360] [color={rgb, 255:red, 0; green, 0; blue, 0 }  ][line width=0.75]    (10.93,-3.29) .. controls (6.95,-1.4) and (3.31,-0.3) .. (0,0) .. controls (3.31,0.3) and (6.95,1.4) .. (10.93,3.29)   ;
\draw    (15,247) -- (444,246) ;
\draw    (15,98) -- (15,247) ;
\draw    (15,98) -- (33,98) ;
\draw [shift={(35,98)}, rotate = 180] [color={rgb, 255:red, 0; green, 0; blue, 0 }  ][line width=0.75]    (10.93,-3.29) .. controls (6.95,-1.4) and (3.31,-0.3) .. (0,0) .. controls (3.31,0.3) and (6.95,1.4) .. (10.93,3.29)   ;
\draw    (299,116) -- (299,247) ;
\draw  [fill={rgb, 255:red, 0; green, 0; blue, 0 }  ,fill opacity=1 ] (581.5,97) .. controls (581.5,93.41) and (584.41,90.5) .. (588,90.5) .. controls (591.59,90.5) and (594.5,93.41) .. (594.5,97) .. controls (594.5,100.59) and (591.59,103.5) .. (588,103.5) .. controls (584.41,103.5) and (581.5,100.59) .. (581.5,97) -- cycle ;
\draw  [fill={rgb, 255:red, 0; green, 0; blue, 0 }  ,fill opacity=1 ] (292.5,247) .. controls (292.5,243.41) and (295.41,240.5) .. (299,240.5) .. controls (302.59,240.5) and (305.5,243.41) .. (305.5,247) .. controls (305.5,250.59) and (302.59,253.5) .. (299,253.5) .. controls (295.41,253.5) and (292.5,250.59) .. (292.5,247) -- cycle ;
\draw  [fill={rgb, 255:red, 208; green, 2; blue, 27 }  ,fill opacity=0.5 ] (217.38,54.62) .. controls (216.37,47.12) and (219.67,39.7) .. (225.89,35.5) .. controls (232.1,31.3) and (240.13,31.07) .. (246.58,34.89) .. controls (248.86,30.53) and (253.04,27.52) .. (257.85,26.77) .. controls (262.65,26.02) and (267.53,27.62) .. (271,31.08) .. controls (272.94,27.13) and (276.76,24.47) .. (281.09,24.06) .. controls (285.43,23.64) and (289.67,25.52) .. (292.31,29.04) .. controls (295.83,24.85) and (301.41,23.08) .. (306.66,24.51) .. controls (311.91,25.93) and (315.87,30.29) .. (316.83,35.7) .. controls (321.14,36.89) and (324.72,39.91) .. (326.66,43.99) .. controls (328.6,48.07) and (328.71,52.8) .. (326.95,56.97) .. controls (331.19,62.56) and (332.18,70.01) .. (329.55,76.54) .. controls (326.93,83.07) and (321.08,87.69) .. (314.19,88.69) .. controls (314.14,94.82) and (310.83,100.44) .. (305.52,103.39) .. controls (300.22,106.34) and (293.76,106.16) .. (288.62,102.92) .. controls (286.44,110.25) and (280.28,115.66) .. (272.82,116.78) .. controls (265.36,117.91) and (257.92,114.56) .. (253.73,108.19) .. controls (248.59,111.33) and (242.42,112.24) .. (236.62,110.7) .. controls (230.81,109.17) and (225.86,105.32) .. (222.88,100.03) .. controls (217.63,100.65) and (212.55,97.89) .. (210.16,93.12) .. controls (207.78,88.35) and (208.6,82.58) .. (212.21,78.68) .. controls (207.53,75.88) and (205.14,70.33) .. (206.29,64.92) .. controls (207.44,59.52) and (211.87,55.48) .. (217.27,54.91) ; \draw   (212.21,78.67) .. controls (214.43,79.99) and (216.98,80.59) .. (219.53,80.39)(222.88,100.03) .. controls (223.98,99.89) and (225.06,99.62) .. (226.08,99.2)(253.73,108.19) .. controls (252.95,107.02) and (252.31,105.76) .. (251.8,104.44)(288.62,102.92) .. controls (289.02,101.58) and (289.28,100.2) .. (289.4,98.81)(314.19,88.69) .. controls (314.24,82.17) and (310.59,76.19) .. (304.8,73.34)(326.95,56.97) .. controls (326.01,59.19) and (324.58,61.16) .. (322.76,62.72)(316.83,35.7) .. controls (316.99,36.59) and (317.07,37.5) .. (317.05,38.42)(292.31,29.04) .. controls (291.44,30.08) and (290.72,31.25) .. (290.17,32.5)(271,31.08) .. controls (270.53,32.03) and (270.18,33.04) .. (269.96,34.07)(246.58,34.89) .. controls (247.94,35.7) and (249.2,36.68) .. (250.33,37.79)(217.38,54.62) .. controls (217.51,55.65) and (217.73,56.68) .. (218.03,57.67) ;
\draw    (191,87) -- (209,87) ;
\draw  [dash pattern={on 4.5pt off 4.5pt}]  (223,116) -- (210.33,139) ;
\draw  [fill={rgb, 255:red, 74; green, 144; blue, 226 }  ,fill opacity=0.5 ] (380,57) -- (563,57) -- (563,127) -- (380,127) -- cycle ;
\draw   (36,52) -- (191,52) -- (191,118) -- (36,118) -- cycle ;
\draw   (444,229) -- (524,229) -- (524,263) -- (444,263) -- cycle ;

\draw (341,126.4) node [anchor=north west][inner sep=0.75pt]    {$x( k)$};
\draw (579,67.4) node [anchor=north west][inner sep=0.75pt]    {$x( k+1)$};
\draw (343,61.4) node [anchor=north west][inner sep=0.75pt]    {$u( k)$};
\draw (238,59.33) node [anchor=north west][inner sep=0.75pt]   [align=left] {Network };
\draw (456,65) node [anchor=north west][inner sep=0.75pt]   [align=left] {Plant};
\draw (387,93.4) node [anchor=north west][inner sep=0.75pt]    {$\Scale[0.9]{x(k+1)=Ax(k) +Bu(k)}$};
\draw (43,60) node [anchor=north west][inner sep=0.75pt]   [align=left] {Feedback controller};
\draw (70,88.4) node [anchor=north west][inner sep=0.75pt]    {$\Scale[0.9]{u(k) =Kx(k)}$};
\draw (455,237) node [anchor=north west][inner sep=0.75pt]   [align=left] {Measure};
\draw (96,149.2) node [anchor=north west][inner sep=0.75pt]   [align=left] {\begin{minipage}[lt]{113.32pt}\setlength\topsep{0pt}
\begin{center}
The network may fail
\end{center}

\end{minipage}};

\end{tikzpicture}
}

%% file: content/6_conclusions.tex
\section{Conclusions}

In this paper, we extended the scope of data-driven stability analysis of hybrid systems. In particular, we considered a more general model of switching linear systems, namely constrained switching linear systems. Based on white-box stability analysis tools and chance-constrained optimization results, we proposed a method that approximates the CJSR of an observed constrained system in two different frameworks. In the first one, we consider that one has access to hybrid states, that is couples of continuous states and edges in the labeled graph constraining the system. Then, we showed that, even if one does not have access to the full state-space of the hybrid system, one can still provide stability guarantees by considering the joint spectral radius of an auxiliary system. Finally, we showcased our method on a real application of CSLS, namely Networked Control Systems.

For further work, we would like to extend our method to even broader families of dynamical systems, for example stochastic switching systems such as Markovian Jump Linear Systems. Second, other prior knowledge on the noise noise could be investigated, such as Gaussian noise for instance. Finally, we plan to relax some assumptions present in this work such as, for example, the amount of knowledge about the labeled graph required to compute the bounds on the CJSR.

%% file: content/7_appendix.tex
\appendix
\section{Proof of Proposition~\ref{non_degenerate_for_sure}} \label{app}

First, we will need the following lemma.
\begin{lemma}[\citep{zero_set}] \label{zero_set_lemma}
	For $x \in \mathbb{R}^n$, let $p(x)$ be a non-zero polynomial. Then, the zero-set $\{x \in \mathbb{R}^n \,|\, p(x) = 0\}$ has zero Lebesgue measure.
\end{lemma}

We note the samples $\omega_N = \{\delta_1, \dots, \delta_N\}$. Moreover, we say that a subset $\beta \subset \omega_N$ is an \emph{essential set} for $\omega_N$ if $(\lambda^*(\beta), c^*(\beta)) = (\lambda^*(\omega_N), c^*(\omega_N))$. Since the objective of the optimization problem \eqref{scenario_program_1} is strongly convex, there exists unique $\lambda$ and $\{P_u\}_{u \in \mathcal{U}}$ such that, for all essential sets $\beta$, the following holds for all $(x, e, w) \in \beta$:
\begin{equation} \label{equality_deg}
	(A_{\sigma^{(l)}(e)}x + w)^T P_{t(e)} (A_{\sigma^{(l)}(e)}x + w) = \lambda^l x^T P_{s(e)} x.
\end{equation}
Denote $I$ the set of indices $i$ in $\omega_N$ such that \eqref{equality_deg} holds for $\delta_i = (x, e, w)$. Now, let us look at the probability to sample a degenerate sample set\footnote{The set of \emph{degenerate} sample sets refers to the complement of the set of non-degenerate sample sets, as defined in Definition~\ref{non_degenerate_property}.}. Since the optimization problem always has at least one tight constraint, if the sample set is degenerate, it must have at least two tight constraints. Then, by symmetry, this probability is
\begin{equation} \label{inequality_degenerate}
\begin{aligned} 
	&\sum_{i = 2}^{N} \begin{pmatrix} N \\ i \end{pmatrix} \pi_l^N(\{\omega_N \, | \, \omega_N \text{ is degenerate and } I = [i] \})  \\
	\leq &\sum_{i = 2}^{N} i\begin{pmatrix} N \\ i \end{pmatrix} \pi_l^N(\{\omega_N \, | \, I = [i] \text{ and } \{\delta_j\}_{j \in [i-1]} \text{ is an essential set for }\omega_N  \}).
\end{aligned}
\end{equation}
Indeed, if the sample set $\omega_N$ is degenerate and $I = [i]$, one of the constraints can be removed from $I$ without changing the solution. Let 
\begin{equation}
	p_i = \pi_l^N(\{\omega_N \, | \, I = [i] \text{ and } \{\delta_j\}_{j \in [i-1]} \text{ is an essential set for }\omega_N  \}), 
\end{equation}
we will show by contradiction that $p_i = 0$ for $i = 2, \dots, N$. Suppose that $p_i > 0$, then we have
\begin{equation} \label{non_zero_samples}
	\pi_l(\{(x, e, w) \in \Delta \,|\, (A_{\sigma^{(l)}(e)}x + w)^T P_{t(e)} (A_{\sigma^{(l)}(e)}x + w) = \lambda^l x^TP_{s(e)}x \}) > 0.
\end{equation}
Now, for any $e \in \mathcal{E}$, we define the polynomial
\begin{equation} \label{quad_pol}
	p_e(x, w) = 
	\begin{pmatrix}
		x^T & w^T
	\end{pmatrix}
	\begin{pmatrix}
		A_{\sigma^{(l)}(e)}^TP_{t(e)}A_{\sigma^{(l)}(e)} - \lambda^l P_{s(e)} & A_{\sigma^{(l)}(e)}^TP_{t(e)} \\
		P_{t(e)}A_{\sigma^{(l)}(e)} & P_{t(e)}
	\end{pmatrix}
	\begin{pmatrix}
		x \\ w
	\end{pmatrix}, 
\end{equation}
and let $\mu = \mu^x \times \mu^w$. Since $\mathcal{E}$ is finite, \eqref{non_zero_samples} implies that there exists $e \in \mathcal{E}$ such that $\mu(\{(x, w) \in \mathbb{S}^n \times B^n[W] \,|\, p_e(x, w) = 0 \}) > 0$.
Now, note that if $(x, w)$ is a root of the polynomial $p_e$, so is $(\alpha x, \alpha w)$, for $\alpha > 0$. Therefore, for all $\alpha > 0$, $\mu(\{(x, w) \in \alpha \mathbb{S}^n \times \alpha B^n[W] \,|\, p(x, w) = 0\}) > 0$.
This implies that the set $\{(x, w) \in \mathbb{R}^n \times \mathbb{R}^n \,|\, p(x, w) = 0\}$ has non-zero Lebesgue measure, since so is its projection in the $x$-coordinate (thanks to the multiplication by $\alpha$), and its projection on the $w$-coordinate (by Assumption~\ref{ass_on_w_dis}). Therefore, by Lemma~\ref{zero_set_lemma}, the polynomial $p_e(x, w)$
is identically zero, which implies $P_{t(e)} = 0$ following Equation~\ref{quad_pol}. However, since $P_{t(e)}$ is a solution of the program \eqref{scenario_program_1}, $P_{t(e)} \succeq I$ following constraint \eqref{PgeqI}, which yields a contradiction. Therefore, for any $i = 1, \dots, N$, $p_i = 0$, and the proof is completed following inequality \eqref{inequality_degenerate}. \qed